\documentclass[10pt]{article}
\usepackage{bm,amsfonts,amssymb,amsthm,amsmath}
\usepackage{mathrsfs}
\usepackage{graphicx,color}
\usepackage{CJK}
\usepackage{indentfirst,enumerate}
\usepackage{multirow}

\DeclareMathAlphabet{\mathpzc}{OT1}{pzc}{m}{it}

\newcommand{\md}{\mathrm{d}}

\newcommand{\avrg}[1]{\hat{#1}}

\newtheorem{theorem}{Theorem}
\numberwithin{theorem}{section}
\newtheorem{lemma}[theorem]{Lemma}
\newtheorem{proposition}[theorem]{Proposition}
\newtheorem{corollary}[theorem]{Corollary}

\numberwithin{equation}{section}

\definecolor{orange}{rgb}{1,0.5,0}
\definecolor{rb}{rgb}{1,0,1}
\allowdisplaybreaks[4]

\allowdisplaybreaks[4]

\begin{document}
\title{\textbf{Quasi-entropy by log-determinant covariance matrix and application to liquid crystals}\footnote{This work is partially supported by NSFC No. 11688101, NCMIS, AMSS grants for Outstanding Youth Project, and ICMSEC dirctor funds.}}
\author{Jie Xu\footnote{LSEC \& NCMIS, Institute of Computational Mathematics and Scientific/Engineering Computing (ICMSEC), Academy of Mathematics and Systems Science (AMSS), Chinese Academy of Sciences, Beijing 100190, China. Email: xujie@lsec.cc.ac.cn}}
\date{}
\maketitle

\begin{abstract}
  A quasi-entropy is constructed for tensors averaged by a density function on $SO(3)$ using the log-determinant of a covariance matrix. 
  It serves as a substitution of the entropy for tensors derived from a constrained minimization that involves integrals. 
  The quasi-entropy is an elementary function that possesses the essential properties of the original entropy.
  It constrains the covariance matrix to be positive definite, is strictly convex, and is invariant under rotations. 
  Moreover, when reduced by symmetries, it keeps the vanishing tensors of the symmetry zero.
  Explicit expressions are provided for axial symmetries up to four-fold, as well as tetrahedral and octahedral symmetries.
  The quasi-entropy is utilized to discuss phase transitions in several systems.
  The results are consistent with using the original entropy.
  Besides, some novel results are presented. 
\end{abstract}

\section{Introduction}
The free energy for a physical system always has an entropy term. 
When the system is described by a density function $\rho$, the entropy term is given as a multiple of the integral of $\rho\ln \rho$ on the space where $\rho$ is defined. 
In many systems, however, using the density function might be too complicated to solve. 
To obtain a simplified description, one might prefer to introduce some quantities averaged by $\rho$. 
Such an approach is commonly adopted in liquid crystals, and the averaged quantities serve as order parameters. 
It then requires to construct an entropy term for these averaged quantities. 
Let us explain by rod-like molecules. 
The density fucntion $\rho$ now depends on a unit vector $\bm{m}$ indicating the direction of the molecule.
A second order symmetric traceless tensor ($3\times 3$ matrix) $Q$ is defined as the average of $\bm{m}\otimes\bm{m}-\mathfrak{i}/3$ about $\rho$, where $\mathfrak{i}$ is the $3\times 3$ identity matrix. 
Apart from the entropy term, the rest of free energy can be expressed in terms of $Q$, usually a polynomial about $Q$ and its derivatives \cite{GennesPierreGillesde1993Tpol}. 
Therefore, to completely write down the free energy, it is necessary to give the entropy term by a function of $Q$.

One simple approach is to assume that we already have a function of $Q$ for the entropy, then to expand it into a polynomial of $Q$ keeping only a few low order terms. 
As a result, it merges into the rest of free energy which is also a polynomial, leading to a Landau-type energy. 
If we ignore the derivative terms, the simplest free energy would be a fourth degree polynomial, 
\begin{equation}
  a\sum_{i_1,i_2}Q_{i_1i_2}Q_{i_1i_2}-b\sum_{i_1,i_2,i_3}Q_{i_1i_2}Q_{i_2i_3}Q_{i_3i_1}+c(\sum_{i_1,i_2}Q_{i_1i_2}Q_{i_1i_2})^2. \label{fepol0}
\end{equation}
Another approach is the so-called Bingham closure, which was probably first proposed in \cite{bingham1974antipodally} and is used widely in simulations \cite{grosso2000closure,feng1998closure,wang1997comparative,Ilg2,INAC}. 
In the simplest case, the free energy consists of an entropy term and a quadratic term, 
\begin{equation}
  f_{\mathrm{ent}}(Q)-a'\sum_{i_1,i_2}Q_{i_1i_2}Q_{i_1i_2}, \label{entfe0}
\end{equation}
where the entropy term $f_{\mathrm{ent}}(Q)$ is derived by solving the constrained minimization problem 
\begin{equation}
  f_{\mathrm{ent}}=\mathrm{min}\int_{S^2}\rho\ln\rho\md\bm{m},\quad \mbox{ s.t. }\int_{S^2}\rho(\bm{m})\md\bm{m}=1,\ \int_{S^2}(\bm{m}\otimes\bm{m}-\frac{\mathfrak{i}}{3})\rho(\bm{m})\md\bm{m}=Q. \label{entmin0}
\end{equation}
It can be shown that if $Q$ is symmetric traceless with its eigenvalues lying in the physically meaningful interval $(-1/3,2/3)$, then $\rho$ is uniquely determined by $Q$ \cite{ball2010nematic}, so that $f_{\mathrm{ent}}$ can be regarded as a function of $Q$. 
Since $f_{\mathrm{ent}}$ is derived from $\int\rho\ln\rho\md\bm{m}$ in a definite setting, we would call $f_{\mathrm{ent}}(Q)$ given by \eqref{entmin0} the original entropy of $Q$. 

Liquid crystals can also be formed by rigid non-spherical molecules other than rod-like molecules, such as bent-core, cross-like and polygon-like molecules \cite{JJAP,PhysRevE.58.5873,Zhao29092015,jakli2018physics}. 
The density function $\rho$ is now a function on $SO(3)$. 
When choosing the averaged quantities for order parameters, multiple tensors are necessarily included. 
Based on various choices of tensors, many different free energies are discussed \cite{Bi1,fel1995tetrahedral,PhysRevE.58.5873,doi:10.1063/1.1649733,roy1999longitudinal,lubensky2002theory,bisi2006universal,brand2010macroscopic,luckhurst2012molecular,trojanowski2012tetrahedratic,shamid2014predicting,gaeta2016octupolar}. 
The choice of tensors and how they are coupled outside the entropy term actually depend on the molecular symmetry, which have been discussed systematically in some recent works \cite{nissinen2016classification,turzi2017determination,xu_softmatter,xu_tensors,xu_kernel}. 
The coefficients of each coupling can also be evaluated if the molecular interaction is specified \cite{RodModel,BentModel}. 

Let us focus on the entropy term. 
The free energy in most works mentioned above is Landau-type, i.e. the entropy term is taken as a polynomial. 
When there are multiple tensors, intrinsic difficulties arise in Landau-type theories. 
Because of stability requirements, the polynomial is at least fourth degree.
The result is that one might have to include too many terms. 
For example, a fourth degree polynomial about two second order tensors, which is used to discuss biaxial liquid crystals, possesses up to fourteen terms \cite{de2008landau}. 
The number of terms increases drastically if more tensors are involved, and it becomes a trouble to find out reasonable values of many coefficients, making it nearly impossible to write down a suitable function to act as the entropy.  
Another problem is that Landau-type theories are not able to constrain the tensors within the physically meaningful domain. 
This can be recognized from $Q$ for rod-like molecules. 
Obviously, a polynomial cannot constrain its eigenvalues in the interval $(-1/3,2/3)$. 
This might not be a major problem for the tensor $Q$. 
However, when there are multiple tensors, the relation between the tensors can result in complicated constraints that must be dealt with carefully.

In contrast, when the entropy term is derived from a minimization problem like \eqref{entmin0}, the above shortcomings can be overcome. 
For the tensor $Q$, the function $f_{\mathrm{ent}}(Q)$ does not have undetermined coefficients. 
Moreover, $Q$ is constrained in the physical domain since it is calculated from a density $\rho$. 
These two advantages also hold when more tensors are involved. 
The price paid for these advantages is that the entropy term involves integration (on $S^2$ in \eqref{entfe0}, on $SO(3)$ in general cases), thus it is difficult in analysis and costly in numerical simulations. 
It takes quite a lot of efforts to obtain analytical results \cite{liu2005axial,fatkullin2005critical,xu2017transmission}. 
Numerically, several fast algorithms are proposed in the case of rod-like molecules, for $f_{\mathrm{ent}}(Q)$ \cite{kent1987asymptotic,kume2005saddlepoint,kume2013saddlepoint,wang2008crucial,luo2018a}, but they cannot be easily extended to the cases with multiple tensors.

We propose a substitution of the original entropy $f_{\mathrm{ent}}$ with an \emph{elementary function} of averaged tensors on the space $SO(3)$, which we call a quasi-entropy. 
The construction is based on a covariance matrix in which all the components of these averaged tensors appear. 
This covariance matrix shall be positive definite. 
To guarantee this, we use the minus log-determinant of the covariance matrix as the quasi-entropy.
Generally, the covariance matrix contains averaged tensors not chosen as order parameters.
To obtain a function about the chosen tensors, the minus log-determinant is minimized with the values of the chosen tensors fixed. 
By using such an elementary function, we can avoid calculating the integral in $f_{\mathrm{ent}}$. 
Moreover, the quasi-entropy is able to keep some essential properties of the original entropy: 
\begin{itemize}
\item Strict convexity. 
\item Invariance under rotations. 
\item Consistency in the reduction by symmetry. Under certain molecular symmetry, it is known that some tensors will vanish when averaged. Consider all the vanishing tensors involved in the covariance matrix. When solving the minimization for the minus log-determinant as we described above, we can show that at the minimizer these tensors are indeed zero. 
\item In the case of rod-like molecules, the asymptotic behaviors are identical to $f_{\mathrm{ent}}(Q)$ given in \eqref{entmin0}. 
\end{itemize}



To illustrate the construction and properties of the quasi-entropy, we introduce some notations and results in Section \ref{prel}. 
The properties of the original entropy are presented in Section \ref{oent}, 
followed by the discussion of quasi-entropy in Section \ref{qentsec}. 
For several symmetries where the tensors involved are no more than fourth order, the explicit expressions are derived in Section \ref{qentexpl}.

It is expected that proposing the quasi-entropy that maintains the essential properties of the original entropy would be able to capture the correct physics while keeping the model simple, especially for phase transitions. 
Indeed, we illustrate it in Section \ref{appl} by investigating important cases, including the consistency with some flagship results when using the original entropy, as well as some novel results. 
The quasi-entropy is combined with a polynomial to give the free energy for liquid crystals formed by molecules with different symmetries. 
The focus is put on the spatially homogeneous phases, and the cases below are examined: 
\begin{itemize}
\item Rod-like molecules, with one second order tensor $Q$ as the order parameter.  We will show that at stationary points, $Q$ must have two equal eigenvalues. This result is consistent with the free energy including the original entropy \eqref{entmin0} \cite{liu2005axial,fatkullin2005critical}. 
\item Molecules of the symmetry of a rectangle/cuboid, with two second order tensors as the order parameters. When the entropy is obtained like in \eqref{entmin0}, some analytical results are given in \cite{xu2017transmission}. 
\item Bent-core molecules, with three order parameter tensors. Armed with the coefficients derived from molecular parameters (see \cite{BentModel}), a phase diagram is presented about molecular parameters. 
  The phase diagram turns out to be very close to that with the original entropy. 
\item Molecules of tetrahedral or octahedral symmetry, with a third order tensor and a fourth order tensor as the order parameters. We will present a phase diagram for a simplified case. To the knowledge of the author, the tetrahedral and octahedral symmetries have not been examined together previously. 
\end{itemize}
From these results, we can see that the quasi-entropy is able to effectively play the role of entropy, meanwhile much simpler than the original entropy. 

Concluding remarks with open problems are given in Section \ref{concl}. 

\section{Preliminaries\label{prel}}

We are considering a system consisting of rigid, non-spherical molecules. 
To represent the orientation of a molecule, we fix a right-handed orthonormal frame $(\bm{m}_1,\bm{m}_2,\bm{m}_3)$ on it. 
The body-fixed frame can rotate with the molecule.
Let us choose in $\mathbb{R}^3$ a right-handed reference orthonormal frame $(\bm{e}_1,\bm{e}_2,\bm{e}_3)$. 
Define a matrix $\mathfrak{p}\in SO(3)$, such that its $(i,j)$ components are given by the $i$-th coordinates of $\bm{m}_j$. 
A parameterization of $\mathfrak{p}$ is given by three Euler angles, 
\begin{align}
\mathfrak{p}=&(\bm{m}_1,\bm{m}_2,\bm{m}_3)
=\left(\begin{array}{ccc}
m_{11} & m_{21} & m_{31}\\
m_{12} & m_{22} & m_{32}\\
m_{13} & m_{23} & m_{33}
\end{array}\right)
\nonumber\\
=&\left(
\begin{array}{ccc}
 \cos\alpha &\ -\sin\alpha\cos\gamma &\ \sin\alpha\sin\gamma\\
 \sin\alpha\cos\beta &\ \cos\alpha\cos\beta\cos\gamma-\sin\beta\sin\gamma &
 \  -\cos\alpha\cos\beta\sin\gamma-\sin\beta\cos\gamma\\
 \sin\alpha\sin\beta &\ \cos\alpha\sin\beta\cos\gamma+\cos\beta\sin\gamma &
 \ -\cos\alpha\sin\beta\sin\gamma+\cos\beta\cos\gamma
\end{array}
\right).\label{EulerRep}
\end{align}

The element $\mathfrak{p}\in SO(3)$ also defines a rotation in the reference frame $(\bm{e}_1,\bm{e}_2,\bm{e}_3)$: for any vector $\bm{a}$, the rotated vector is given by $\mathfrak{p}\bm{a}$. 
In this way, we could view the product $\mathfrak{tp}$ of two elements in $SO(3)$ as the rotation of the frame $\mathfrak{p}=(\bm{m}_1,\bm{m}_2,\bm{m}_3)$ by $\mathfrak{t}$ by imposing the rotation $\mathfrak{t}$ on the three $\bm{m}_i$. 
The multiplication on the right, $\mathfrak{pt}$, has a different meaning. 
It can be regarded as performing the rotation with the body-fixed frame $(\bm{m}_1,\bm{m}_2,\bm{m}_3)$ taken as the reference frame. 
To comprehend the multiplications on the left and on the right, let us consider 
\begin{equation}
  \mathfrak{j}_{\theta}=\left(
  \begin{array}{ccc}
    1 & 0 & 0\\
    0 & \cos\theta & -\sin\theta\\
    0 & \sin\theta & \cos\theta
  \end{array}
  \right). \label{jtheta}
\end{equation}
The product $\mathfrak{j}_{\theta}\mathfrak{p}$ keeps the first line of $\mathfrak{p}$ invariant, 
while the product $\mathfrak{pj}_{\theta}$ keeps the first column, i.e. $\bm{m}_1$, invariant.
We also point out a simple observation to be used later: if the Euler angles of $\mathfrak{p}$ are $(\alpha,\beta,\gamma)$, then those of $\mathfrak{pj}_{\theta}$ are $(\alpha,\beta,\gamma+\theta)$. 

Using the Euler angles, the uniform probability measure on $SO(3)$ is expressed as $\md\mathfrak{p}=(1/8\pi^2)\sin\alpha\md\alpha\md\beta\md\gamma$. 
It is invariant under two kinds of rotations: for any function $g(\mathfrak{p})$ and $\mathfrak{t}\in SO(3)$, we have 
\begin{equation}
  \int g(\mathfrak{tp})\,\md\mathfrak{p}=\int g(\mathfrak{pt})\,\md\mathfrak{p}=\int g(\mathfrak{p})\,\md\mathfrak{p}. \label{introt}
\end{equation}

Next, we introduce some notations on tensors.
A $k$-th order tensor can be written as a linear combination of the basis given by the tensor product of $\bm{e}_i$, 
\begin{equation}
  U=\sum_{j_1,\ldots,j_k=1}^3U_{j_1\ldots j_k}\bm{e}_{j_1}\otimes\ldots\otimes\bm{e}_{j_k}, 
\end{equation}
where $U_{j_1\ldots j_k}$ are the components. 
The dot product is given by summing up the product of corresponding components, 
\begin{equation}
  U\cdot V= \sum_{j_1,\ldots,j_k=1}^3U_{j_1\ldots j_k}V_{j_1\ldots j_k},\quad |U|^2=U\cdot U. 
\end{equation}
A rotation $\mathfrak{p}$ acting on a tensor is defined by replacing $\bm{e}_i$ by $\bm{m}_i$ without changing the components: 
\begin{align}
  \mathfrak{p}\circ U=&\sum_{j_1,\ldots,j_k=1}^3U_{j_1\ldots j_n}\bm{m}_{j_1}\otimes\ldots\otimes\bm{m}_{j_n}.\label{rot0}
\end{align}
It gives a function of $\mathfrak{p}$, which we also denote by $U(\mathfrak{p})$. 
It is easy to verify that $U(\mathfrak{p}_1\mathfrak{p}_2)=\mathfrak{p}_1\circ U(\mathfrak{p}_2)$, and 
\begin{equation}
  U_1(\mathfrak{sp}_1)\cdot U_2(\mathfrak{sp}_2)=U_1(\mathfrak{p}_1)\cdot U_2(\mathfrak{p}_2), \quad\forall \mathfrak{s}\in SO(3). \label{innrot}
\end{equation}
Any tensor rotated by the identity $\mathfrak{i}$ is just itself, i.e. $U(\mathfrak{i})=U$. 

A $k$-th order tensor $U$ is called symmetric, if for any permutation $\sigma$ of $(1 \ldots k)$ it satisfies $U_{j_{\sigma(1)}\ldots j_{\sigma(k)}}=U_{j_1\ldots j_k}$.
For any $k$-th order tensor $U$, we define its symmetrization, 
\begin{equation}
  U_{\mathrm{sym}}=\frac{1}{k!}\sum_{j_1,\ldots,j_k=1}^3\sum_{\sigma}U_{j_{\sigma(1)}\ldots j_{\sigma(k)}}\bm{e}_{j_1}\otimes\ldots\otimes\bm{e}_{j_k}, \label{Usym}
\end{equation}
where the summation inside is taken over all the permutations. 
The trace of a $k$-th order symmetric tensor $U$ is defined by contracting two of its components, giving a $(k-2)$-th order symmetric tensor, 
\begin{align*}
  (\mathrm{tr}U)_{j_1\ldots j_{k-2}}=\sum_{i=1}^3U_{j_1\ldots j_{k-2}ii}. 
\end{align*}
A symmetric tensor $U$ is called a symmetric traceless tensor if $\mathrm{tr}U=0$. 
It is easy to verify that the symmetric and traceless properties are kept under rotations, i.e. if $U$ is symmetric (traceless), then $U(\mathfrak{p})$ is also symmetric (traceless). 

To express symmetric tensors, we introduce the monomial notation below, 
\begin{equation}
\bm{m}_1^{k_1}\bm{m}_2^{k_2}\bm{m}_3^{k_3}
=(\underbrace{\bm{m}_1\otimes\ldots}_{k_1}\otimes
\underbrace{\bm{m}_2\otimes\ldots}_{k_2}\otimes
\underbrace{\bm{m}_3\otimes\ldots}_{k_3}
)_{\mathrm{sym}}. \label{tensor_monomial}
\end{equation}
In this way, a polynomial about $\bm{m}_i$ can be regarded as a symmetric tensor, if every term in the polynomial has the same order. 
The $3\times 3$ identity matrix $\mathfrak{i}$ can be expressed as a polynomial, 
\begin{equation}
  \mathfrak{i}=\bm{m}_1^2+\bm{m}_2^2+\bm{m}_3^2. \label{squaresum}
\end{equation}

We use the Kronecker delta and the Levi-Civita symbol, which are given by 
$$
\delta_{ij}=\left\{
\begin{array}{ll}
  1, &i=j, \\
  0, &i\ne j.
\end{array}
\right.\quad
\epsilon_{ijk}=\left\{
\begin{array}{ll}
  1, &(ijk)=(123),(231),(312), \\
  -1, &(ijk)=(132),(213),(321),\\
  0, &\text{otherwise}.
\end{array}
\right.
$$

In the following, we present some results about tensors established in previous works, and some corollaries derived from these results. 
There might be many useful references, such as \cite{SpecFun,Group_Boerner,jerphagnon1978the,turzi2011on}. 
Here, we mainly follow \cite{xu_tensors} since the notations are identical, where all the propositions in this section can be found. 
It is known that the $k$-th order symmetric traceless tensors form a linear space of the dimension $2k+1$. 
Thus, let us choose an orthongonal basis which we denote as 
\begin{equation}
  \{W^k_1,\ldots,W^k_{2k+1}\},\text{ s.t. } W^k_i\cdot W^k_j=\delta_{ij}. \label{orth_basis}
\end{equation}
These tensors have the superscript $k$ to indicate their tensor order, which will be frequently adopted whenever we would like to emphasize the tensor order. 
These tensors also form an orthogonal basis when rotated by some $\mathfrak{p}\in SO(3)$. 
Any $k$-th order symmetric traceless tensor can be linearly expressed by the basis, 
\begin{equation}
  U^k(\mathfrak{p})=\sum_{j}u_j
  W^k_j(\mathfrak{p}). \label{linexpr0}
\end{equation}
The coefficients are \emph{independent of} $\mathfrak{p}$ because of the definition \eqref{rot0}. 
Thus, the transformation between two basis $W^k_j(\mathfrak{ps})$ and $W^k_j=W^k_j(\mathfrak{p})$ can be expressed as 
\begin{equation}
  W^k_j(\mathfrak{ps})=\sum_{j'=1}^{2k+1}\Gamma^k_{jj'}(\mathfrak{s})W^k_{j'}(\mathfrak{p}), \quad \mathfrak{s}\in SO(3), \label{basistrans}
\end{equation}
where the $(2k+1)\times(2k+1)$ matrix $\Gamma$ only depends on $\mathfrak{s}$. 
It is easy to see that 
\begin{align}
  &\Gamma(\mathfrak{s}_1\mathfrak{s}_1)=\Gamma(\mathfrak{s}_1)\Gamma(\mathfrak{s}_2), \quad \Gamma(\mathfrak{i})=I_{2k+1}, \label{rotmat}
\end{align}
where $I_{m}$ denotes the $m\times m$ identity matrix. 

%

Now, let us consider the components of an $n$-th order tensor $U(\mathfrak{p})$, which are functions of $\mathfrak{p}$. They can be expressed linearly by the components of symmetric traceless tensors, which we state in the proposition below.
\begin{proposition}\label{orth_basis_mat}
  The functions $W^k_i(\mathfrak{i})\cdot W^k_j(\mathfrak{p})$, for $k=0,1,\ldots,n$ and $1\le i,j\le 2k+1$, are linearly independent and can express linearly the components of an $n$-th order tensors $U(\mathfrak{p})$. 
\end{proposition}
\begin{corollary}\label{w2n}
  The functions $W^k_i(\mathfrak{i})\cdot W^k_j(\mathfrak{p})$, for $k=0,1,\ldots,2n$ and $1\le i,j\le 2k+1$, are linearly equivalent to
  \begin{align*}
  &\Big(W^k_i(\mathfrak{i})\cdot W^k_j(\mathfrak{p})\Big)\Big(W^{k'}_{i'}(\mathfrak{i})\cdot W^{k'}_{j'}(\mathfrak{p})\Big),\\ &\qquad 0\le k,k'\le n,\ 1\le i,j\le 2k+1,\ 1\le i',j'\le 2k'+1. 
  \end{align*}
\end{corollary}
\begin{proof}
  The functions $W^k_i(\mathfrak{i})\cdot W^k_j(\mathfrak{p})$ for $k=0,1,\ldots,2n$ are linearly equivalent to the components of $U(\mathfrak{p})$ where $U$ takes all the $2n$-th order tensors. The basis of $2n$-th order tensor can be naturally expressed as the tensor product of two $n$-th order tensor, i.e.
  $$
  \bm{e}_{j_1}\otimes\ldots\otimes\bm{e}_{j_{2n}}=(\bm{e}_{j_1}\otimes\ldots\otimes\bm{e}_{j_{n}})\otimes(\bm{e}_{j_{n+1}}\otimes\ldots\otimes\bm{e}_{j_{2n}}). 
  $$
  Thus, $U(\mathfrak{p})$ can be written as the sum of several tensors of the form $V_1(\mathfrak{p})\otimes V_2(\mathfrak{p})$ where $V_1$, $V_2$ are $n$-th order tensors. Notice that the components of $V_1(\mathfrak{p})$ and $V_2(\mathfrak{p})$ can be expressed linearly by $W^k_i(\mathfrak{i})\cdot W^k_j(\mathfrak{p})$ for $k=0,1,\ldots,n$. 
\end{proof}

Finally, we consider the molecular symmetry. 
A rigid molecule might be invariant under a rotation $\mathfrak{s}\in SO(3)$.
If we look at the frame $\mathfrak{p}=(\bm{m}_1,\bm{m}_2,\bm{m}_3)$, the rotation transforms $\mathfrak{p}$ to $\mathfrak{ps}$. 
All such rotations form a subgroup $\mathcal{G}$ of $SO(3)$. 
According to the group $\mathcal{G}$, the space of $n$-th order symmetric traceless tensors can be decomposed into two subspaces.
Let 
\begin{equation}
  \mathbb{A}^{\mathcal{G},n}=\{n\text{-th order }T(\mathfrak{p}): T(\mathfrak{ps})=T(\mathfrak{p}),\ \forall \mathfrak{s}\in \mathcal{G}\},\label{tgset}
\end{equation}
which is the space consisting of all the $n$-th order symmetric traceless tensors $T$ invariant under any rotation in $\mathcal{G}$. 
It is equivalent to require a tensor invariant under the generating elements of $\mathcal{G}$.
The orthogonal complement of this space is characterized by the following proposition. 
\begin{proposition}\label{inv_van}
  The orthogonal complement of the space $\mathbb{A}^{\mathcal{G},n}$ is given by all the $n$-th order symmetric traceless tensors vanishing when averaged over $\mathcal{G}$, 
  \begin{equation}
    (\mathbb{A}^{\mathcal{G},n})^{\perp}=\left\{T(\mathfrak{p}): \frac{1}{\#\mathcal{G}}\sum_{\mathfrak{s}\in\mathcal{G}}T(\mathfrak{ps})=0\right\}. \label{tgcompset}
  \end{equation}
\end{proposition}
The above decomposition is useful when we examine the tensors averaged by the density function $\rho(\mathfrak{p})$.
If a rigid molecule possesses the symmetry given by the group $\mathcal{G}$, the density function satisfies $\rho(\mathfrak{ps})=\rho(\mathfrak{p})$ for any $\mathfrak{s}\in \mathcal{G}$.
As a result, those tensors in the space $(\mathbb{A}^{\mathcal{G},n})^{\perp}$ are zero when averaged by $\rho(\mathfrak{p})$. 

\section{Entropy of tensors by minimization\label{oent}}
In this section, we discuss the properties of the entropy term about tensors deduced from minimization. 
We start from formulating the minimization problem. 
The state of a system containing rigid molecules can be described by the density function $\rho(\bm{x},\mathfrak{p})$. 
In this work, we do not consider spatial variations, so we assume that $\rho=\rho(\mathfrak{p})$ only depends on the orientation $\mathfrak{p}$. 
For any function $g(\mathfrak{p})$ on $SO(3)$, we denote its average over the density $\rho(\mathfrak{p})$ as $\langle g(\mathfrak{p})\rangle$, i.e. 
\begin{equation}
  \langle g(\mathfrak{p})\rangle = \int_{SO(3)} g(\mathfrak{p})\rho(\mathfrak{p})\md\mathfrak{p}. \label{avrgdef}
\end{equation}

The free energy about $\rho$ is written as 
\begin{equation}
  F[\rho]=\int \rho\ln\rho \md\mathfrak{p}+F_{\mathrm{i}}[\rho]. 
\end{equation}
The second term in the above represents contribution of molecular interactions. Regardless of what molecule we consider, this term can be expanded, and the resulting terms are given by several symmetric traceless tensors averaged by $\rho$ (see \cite{xu_kernel,RodModel,BentModel} and the references therein). 
Using the notation in \eqref{avrgdef}, these averaged tensors are written as $\langle U_j(\mathfrak{p})\rangle$ with emphasis on what $U_j(\mathfrak{p})$ are averaged. 
We also introduce another notation for the same quantity, $\avrg{U}_j=\langle U_j(\mathfrak{p})\rangle$, to represent the values of these averaged tensors under a certain $\rho$.
In most cases, these two notations are the same but with different emphasis. 
However, there will be cases where they are not identical, which we will point out. 
The free energy is now written in the form 
\begin{equation}
  \int \rho\ln\rho \md\mathfrak{p} + F_{\mathrm{i}}[\avrg{U}_1\ldots,\avrg{U}_k]. 
\end{equation}
If we could write the entropy term as a function of $\avrg{U}_j$, we eventually obtain a free energy about these tensors. 
To accompolish this, the entropy is minimized with the values of $\avrg{U}_j$ fixed. 
In other words, we solve the constrained minimization problem below, 
\begin{align}
  \min \int \rho\ln\rho \md\mathfrak{p},\quad \text{ s.t. }\langle 1\rangle=\int\rho(\mathfrak{p}) \md\mathfrak{p}=1,\quad \langle U_j(\mathfrak{p})\rangle=\int U_j(\mathfrak{p})\rho(\mathfrak{p})\md\mathfrak{p}=\avrg{U}_j. \label{entmingen}
\end{align}
For the tensor $Q$, it is exactly the Bingham closure. 
The solution is unique, as the following proposition is given in \cite{xu_tensors}.
\begin{proposition}\label{unique0}
  Assume there exists a $0\le\rho<+\infty$ such that $\langle U_j(\mathfrak{p})\rangle=\avrg{U}_j$. The problem \eqref{entmingen} has a unique solution of the form 
  \begin{equation}
    \rho(\mathfrak{p})=\frac{1}{Z}\exp\Big(\sum_{j=1}^{k}B_j\cdot U_j(\mathfrak{p})\Big), \label{Boltz}
  \end{equation}
  where $B_j$ are symmetric traceless tensors of the same order as $U_j$, and 
  \begin{equation}
    Z=\int \exp\Big(\sum_{j=1}^{k}B_j\cdot U_j(\mathfrak{p})\Big)\md\mathfrak{p}. \label{PartFun}
  \end{equation}
\end{proposition}
Let the density be \eqref{Boltz} in the entropy. It becomes 
\begin{equation}
  f_{\mathrm{ent}}(\avrg{U}_1,\ldots,\avrg{U}_k)=\sum_{j=1}^{k}B_j\cdot \avrg{U}_j-\ln Z, \label{entropy0}
\end{equation}
with $Z$ given in \eqref{PartFun}. 
We write it as a function of $\avrg{U}_j$, because $B_j$ are uniquely determined by $\avrg{U}_j$ according to the proposition, and $Z$ is a function of $B_j$. 
Similarly, we denote the density in \eqref{Boltz} as $\rho\big(\mathfrak{p}|\langle U_j(\mathfrak{p})\rangle=\avrg{U}_j\big)$, or in short as $\rho(\mathfrak{p}|\avrg{U}_j)$ if it is unambiguous what tensors are averaged. 




Now let us show some properties of $f_{\mathrm{ent}}$. 
\begin{enumerate}
\item It is strictly convex about the tensors $\avrg{U}_j$. This comes from the strict convexity of $\rho\ln\rho$. 
For any two sets of tensors $(\avrg{U}_{j,1})$ and $(\avrg{U}_{j,2})$, the density function $\tilde{\rho}=\Big(\rho(\mathfrak{p}|\avrg{U}_{j,1})+\rho(\mathfrak{p}|\avrg{U}_{j,2})\Big)/2$ gives the averaged tensors $\langle U_j(\mathfrak{p})\rangle=(\avrg{U}_{j,1}+\avrg{U}_{j,2})/2$. 
Since $f_{\mathrm{ent}}\Big((\avrg{U}_{j,1}+\avrg{U}_{j,2})/2\Big)$ is the minimum value of the entropy with the tensors taking the values $(\avrg{U}_{j,1}+\avrg{U}_{j,2})/2$, we deduce that 
\begin{align*}
f_{\mathrm{ent}}\Big((\avrg{U}_{j,1}+\avrg{U}_{j,2})/2\Big)\le & \int\tilde{\rho}\ln \tilde{\rho}\md\mathfrak{p}\\
\le & \frac{1}{2}\int\Big(\rho(\mathfrak{p}|\avrg{U}_{j,1})\ln \rho(\mathfrak{p}|\avrg{U}_{j,1})+\rho(\mathfrak{p}|\avrg{U}_{j,2})\ln \rho(\mathfrak{p}|\avrg{U}_{j,2})\Big)\md\mathfrak{p}\\
= & \frac{1}{2}\Big(f_{\mathrm{ent}}(\avrg{U}_{j,1})+f_{\mathrm{ent}}(\avrg{U}_{j,2})\Big). 
\end{align*}
The equality holds if and only if $\rho(\mathfrak{p}|\avrg{U}_{j,1})=\rho(\mathfrak{p}|\avrg{U}_{j,2})$, which is equivalent to $\avrg{U}_{j,1}=\avrg{U}_{j,2}$ by Proposition \ref{unique0}. 
\item It is invariant under rotations in two different ways, as we describe below. 
\begin{enumerate}
\item Rotate the values of the averaged tensors, $\avrg{U}_j$, to $\avrg{U}_j(\mathfrak{t})$. We have 
  $$
  f_{\mathrm{ent}}(\avrg{U}_j)= f_{\mathrm{ent}}\big(\avrg{U}_j(\mathfrak{t})\big).
  $$
\item Rotate the tensors to be averaged, $U_j(\mathfrak{p})$, to $U_j(\mathfrak{ps})$. We actually consider the minimization problem with the constraints 
  \begin{equation}
    \langle U_j(\mathfrak{ps})\rangle=\avrg{U}_j.\nonumber
  \end{equation}
  The entropy $f_{\mathrm{ent}}$ is also invariant under the above rotation. 
\end{enumerate}
To show the invariance under rotations, let us write down the density $\rho$ for both cases. 
For the case (a), consider the density 
\begin{equation}
\rho(\mathfrak{p})=\frac{1}{Z}\exp\big(\sum_j B_j(\mathfrak{t})\cdot U_j(\mathfrak{p})\big). \label{density_rot1}
\end{equation}
We can calculate using \eqref{introt} and \eqref{innrot} that 
\begin{align*}
  Z=&\int \exp\big(\sum_j B_j(\mathfrak{t})\cdot U_j(\mathfrak{p})\big)\,\md\mathfrak{p}
  =\int \exp\big(\sum_j B_j\cdot U_j(\mathfrak{t}^{-1}\mathfrak{p})\big)\,\md\mathfrak{p}\\
  =&\int \exp\big(\sum_j B_j\cdot U_j(\mathfrak{p})\,\md\mathfrak{p}\big)\,\md\mathfrak{p}. 
\end{align*}
Therefore, the average is calculated as 
\begin{align*}
  \langle U_i(\mathfrak{p})\rangle=&\frac{1}{Z}\int \exp\big(\sum_j B_j(\mathfrak{t})\cdot U_j(\mathfrak{p})\big)U_i(\mathfrak{p})\,\md\mathfrak{p}\\
  =&\frac{1}{Z}\int \exp\big(\sum_j B_j\cdot U_j(\mathfrak{t}^{-1}\mathfrak{p})\big)U_i(\mathfrak{p})\,\md\mathfrak{p}\\
  =&\frac{1}{Z}\int \exp\big(\sum_j B_j\cdot U_j(\mathfrak{p})\big)U_i(\mathfrak{tp})\,\md\mathfrak{p}\\
  =&\mathfrak{t}\circ\frac{1}{Z}\int \exp\big(\sum_j B_j\cdot U_j(\mathfrak{p})\big)U_i(\mathfrak{p})\,\md\mathfrak{p}\\
  =&\mathfrak{t}\circ\avrg{U}_i. 
\end{align*}
By Proposition \ref{unique0}, $\rho\big(\mathfrak{p}|\avrg{U}_j(\mathfrak{t})\big)$ is exactly given by \eqref{density_rot1}. So we have
\begin{align*}
  f_{\mathrm{ent}}\big(\avrg{U}_j(\mathfrak{t})\big)=&\sum_{j=1}^{k}B_j(\mathfrak{t})\cdot \avrg{U}_j(\mathfrak{t})-\ln Z
  =\sum_{j=1}^{k}B_j\cdot \avrg{U}_j-\ln Z
  =f_{\mathrm{ent}}\big(\avrg{U}_j). 
\end{align*}
For the case (b), a similar procedure yields 
\begin{equation}
\rho\big(\mathfrak{p}|\langle U_j(\mathfrak{ps})\rangle=\avrg{U}_j\big)=\frac{1}{Z}\exp\big(\sum_j B_j\cdot U_j(\mathfrak{ps})\big). \label{density_rot2}
\end{equation}
The invariance is obvious using this density. 
\item We take molecular symmetry into consideration. By the result of \cite{xu_kernel}, when the molecule allows rotations in a group $\mathcal{G}\subseteq SO(3)$, the tensors $U_j$ can only be the invariant tensors, i.e. $U_j(\mathfrak{ps})=U_j(\mathfrak{p})$ for any $\mathfrak{s}\in\mathcal{G}$.
  As a result, the density $\rho(\mathfrak{p}|\avrg{U}_j)$ satisfies 
  \begin{align}
    \rho(\mathfrak{ps}|\avrg{U}_j)=&\frac{1}{Z}\exp\big(\sum_j B_j\cdot U_j(\mathfrak{ps})\big)
    =\frac{1}{Z}\exp\big(\sum_j B_j\cdot U_j(\mathfrak{ps})\big)
    =\rho(\mathfrak{p}|\avrg{U}_j). \label{density_sym}
  \end{align}
  We could calculate the average of any tensor using the density $\rho$. In particular, for a tensor $U\in(\mathbb{A}^{\mathcal{G},l})^{\perp}$, we have $\langle U(\mathfrak{p})\rangle=0$ because of \eqref{density_sym} and Proposition \ref{inv_van}. 
\end{enumerate}

We would like to formulate the entropy in a different way, which is convenient for us to explain the idea of proposing the quasi-entropy. 
In the above, there is no restriction in \eqref{entmingen} on what tensors $U_j$ can be chosen. 
Let us choose a special set of tensors: all the basis symmetric traceless tensors $W^k_j$ up to $m$-th order, i.e. $k=1,\ldots, m$ and $1\le j\le 2k+1$, for some $m$. 
Here, we do not need to include $W^0_1$, because it is a scalar so that $W^0_1(\mathfrak{p})$ is a constant function. 
By minimizing the entropy under the constraints $\langle W^k_j(\mathfrak{p})\rangle=\avrg{W}^k_j$, we determine a density function $\rho(\mathfrak{p}|\avrg{W}^k_j)$, and obtain the entropy 
\begin{equation}
  f_m=f_{\mathrm{ent}}\Big(\avrg{W}^k_j\big|_{k=1}^m\Big). 
\end{equation}
The subscript $m$ of $f_m$ indicates that all the tensors up to $m$-th order are included. 

Now, we go back to the general choice of $U_j$. 
Denote by $n_j$ the order of $U_j$, and suppose $m$ is no less than the maximum of $n_j$. 
Express $U_j$ by a linear combination of $W^{n_j}_{j'}$, 
\begin{equation}
  U_j(\mathfrak{p})=\sum_{j'=1}^{2n_j+1} u_{jj'}W^{n_j}_{j'}(\mathfrak{p}), \label{coefU}
\end{equation}
where $u_{jj'}$ are independent of $\mathfrak{p}$ (cf. \eqref{linexpr0}). 
Then, the averaged tensors satisfy 
\begin{equation}
  \avrg{U}_j=\sum_{j'=1}^{2n_j+1} u_{jj'}\avrg{W}^{n_j}_{j'}. 
\end{equation}
Using this linear relation, we consider the following minimization problem, 
\begin{align}
  \min_{\avrg{W}^k_j} f_{m},\quad \mbox{s.t. }\sum_{j'=1}^{2n_j+1} u_{jj'}\avrg{W}^{n_j}_{j'}=\avrg{U}_j. \label{minprob1}
\end{align}
By Proposition \ref{unique0}, its unique solution is given by $W^k_j(\mathfrak{p})$ averaged by the density $\rho(\mathfrak{p}|\avrg{U}_j)$. Therefore, the entropy $f_{\mathrm{ent}}(\avrg{U}_j)$ can also be given by the minimum value of $f_m$ subject to some linear constraints. 
The three properties discussed above can be restated as follows. 
\begin{enumerate}
\item Strict convexity about $\avrg{U}_j$. 
\item Invariance under rotations. 
  \begin{enumerate}
  \item If we replace $\avrg{U}_j$ with $\avrg{U}_j(\mathfrak{t})$, the minimum value of \eqref{minprob1} is the same. 
  \item We change the constraints into 
    \begin{equation}
      \sum_{j',j''=1}^{2n_j+1} u_{jj'}\Gamma_{j'j''}(\mathfrak{s})\avrg{W}^{n_j}_{j''}=\avrg{U}_j. 
    \end{equation}
    Then, the minimum value of \eqref{minprob1} is also the same. 
    According to \eqref{basistrans}, the above equation just translates the conditions $\langle U_j(\mathfrak{ps})\rangle=\avrg{U}_j$ into the language of the basis $W^{n_j}_{j'}$. 
  \end{enumerate}
\item Assume that for all $j$, we have $\sum_{j'}u_{jj'}W^{n_j}_{j'}(\mathfrak{p})\in \mathbb{A}^{\mathcal{G},n_j}$. For any $l$-th order symmetric traceless tensor $U(\mathfrak{p})\in(\mathbb{A}^{\mathcal{G},l})^{\perp}$ for $l\le m$, let us express it by $W^l_j(\mathfrak{p})$:
  $$
  U(\mathfrak{p})=\sum_{j=1}^{2l+1}u_jW^l_j(\mathfrak{p}). 
  $$
  Then, the minimizer of \eqref{minprob1} must satisfy $\sum_{j=1}u_j\avrg{W}^l_j=0$. 
\end{enumerate}

A few remarks are given on the formulation \eqref{minprob1}. 
It is noticed that only the average tensors $\avrg{U}_j$ and $\avrg{W}^k_j$ are involved. 
The relations between the averaged tensors are given through the coefficients $u_{jj'}$ that are derived from the unaveraged tensors $U_j(\mathfrak{p})$ and $W^k_j(\mathfrak{p})$. 
In other words, we only use $U_j(\mathfrak{p})$ and $W^k_j(\mathfrak{p})$ to figure out the linear relations between the averaged tensors. 
After the coefficients are obtained, we can forget about what tensors are averaged. 
This idea will also be adopted in the next section when constructing the quasi-entropy. 

  
\section{Quasi-entropy\label{qentsec}}

We know from the previous section that whatever the tensors $U_j$ we choose, we can always write down an entropy $f_{\mathrm{ent}}$ as a function of the averaged values $\avrg{U}_j$. 
However, such a function is implicitly defined where integrals on $SO(3)$ are involved. 
It thus brings huge difficulties in both theoretical analysis and numerical simulations when applying to particular systems. 
In this section, we construct an elementary function about $\avrg{U}_j$ as a substitution of $f_{\mathrm{ent}}$. 
We shall prove that the function satisfies the properties we claimed in the previous section. 

Let us consider the components of $U(\mathfrak{p})$ where $U$ takes all the tensors of the order no greater than certain order $n$. 
They are all functions on $SO(3)$. 
By Proposition \ref{orth_basis_mat}, a basis of these functions can be given by 
\begin{equation}
  W^k_j\cdot W^k_{j'}(\mathfrak{p}),\quad 0\le k\le n,\ 1\le j,j'\le 2k+1. \label{basisfun}
\end{equation}
We have mentioned that $k=0$ gives a constant function, which we will deal with separately. 
Let all the remaining functions (i.e. $1\le k\le n$) form a vector $\bm{w}_n(\mathfrak{p})$, where the $n$ in the subscript represents the highest tensor order we choose. 
Its average is denoted by $\langle\bm{w}_n(\mathfrak{p})\rangle=\avrg{\bm{w}}_n$, where the components are given by $W^k_j\cdot \avrg{W}^k_{j'}$. Similar to the previous section, the notation on the left-hand side emphasizes what function is averaged, and the right-hand side represents the value of the averaged function. 

When these functions are averaged by some density function $\rho>0$, the averaged values $\langle\bm{w}_n(\mathfrak{p})\rangle=\avrg{\bm{w}}_n$ shall meet some constraints. 
For any scalar function $g(\mathfrak{p})$, we have $\langle \big(g(\mathfrak{p})\big)^2\rangle\ge \langle g(\mathfrak{p})\rangle^2$ and the equality holds only if $g$ is a constant function. 
Since we have excluded the constant function out of $\bm{w}_n(\mathfrak{p})$, the covariance matrix, 
\begin{equation}
  \langle\bm{w}_n(\mathfrak{p})\big(\bm{w}_n(\mathfrak{p})\big)^t\rangle-\langle\bm{w}_n(\mathfrak{p})\rangle\langle\bm{w}_n(\mathfrak{p})\rangle^t, \label{covorg}
\end{equation}
shall be positive definite. 
Using Corollary \ref{w2n}, the components of the matrix $\bm{w}_n(\mathfrak{p})\big(\bm{w}_n(\mathfrak{p})\big)^t$ can be expressed linearly by $\bm{w}_{2n}(\mathfrak{p})$ together with the constant function. 
Therefore, we can find some matrices $Y_j$ and write 
\begin{equation}
  \bm{w}_n(\mathfrak{p})\big(\bm{w}_n(\mathfrak{p})\big)^t=\sum_{j=1}^{d(2n)} Y_j w_{2n,j}(\mathfrak{p})+Y_0. 
\end{equation}
Here, we use $d(n)=\sum_{j=1}^n(2j+1)^2$ to denote the length of the vector $\bm{w}_n$, $w_{2n,j}$ denotes the $j$-th component of the vector $\bm{w}_{2n}$, and $Y_j$ for $0\le j\le d(2n)$ are constant matrices of the size $d(n)\times d(n)$. 
According to the same convention, we will use $\avrg{w}_{2n,j}$ to denote the $j$-th component of the vector $\avrg{\bm{w}}_{2n}$. 
Thus, the covariance matrix can be written as 
\begin{equation}
  C(\avrg{\bm{w}}_{2n})=\sum_{j=1}^{d(2n)} Y_j\avrg{w}_{2n,j}+Y_0-\avrg{\bm{w}}_n(\avrg{\bm{w}}_n)^t. \label{covmat}
\end{equation}

As we remarked at the end of the previous section, in \eqref{covmat} only the linear relations between $\bm{w}_n(\mathfrak{p})\big(\bm{w}_n(\mathfrak{p})\big)^t$ and $\bm{w}_{2n}(\mathfrak{p})$ are inherited, in terms of $Y_j$, by the averaged values $\avrg{\bm{w}}_{2n}$. 
As a result, the right-hand side of \eqref{covmat} is an expression depending only on $\avrg{\bm{w}}_{2n}$ without further constraints. 
Hence, we need a function that is able to constrain the matrix $C(\avrg{\bm{w}}_{2n})$ positive definite. 
Taking this requirement into consideration, the quasi-entropy we propose is based on the log-determinant of the covariance matrix, 
\begin{equation}
  q_{2n}^{(0)}(\avrg{W}^k_j|_{k\le 2n})=
    -\ln\det C(\avrg{\bm{w}}_{2n}),\quad C(\avrg{\bm{w}}_{2n})\text{ positive definite}. 
  \label{logdet_general}
\end{equation}
We write it as a function of $\avrg{W}^k_j$ since $\avrg{\bm{w}}_{2n}$ gives all the linearly independent components. 
The function $q_{2n}^{(0)}$ is defined only for $\avrg{\bm{w}}_{2n}$ such that $C(\avrg{\bm{w}}_{2n})$ is positive definite.
It is easy to see that for any sequence of positive definite matrix $C_j$, the limit of $j\to +\infty$, if it exists, is positive semi-definite. 
If the limit is a singular matrix, we have 
\begin{equation}
  \lim_{j\to +\infty}-\ln\det C_j=+\infty. \nonumber
\end{equation}
Thus, the function $q_{2n}^{(0)}$ actually gives a barrier function that enforces the positive-definiteness. 

The vector $\bm{w}_n(\mathfrak{p})$ has multiple choices. 
We shall show that the resulting quasi-entropy is independent of how $\bm{w}_n(\mathfrak{p})$ is chosen. 
\begin{proposition}\label{anotherw}
  Choose any vector $\bm{w}'_n(\mathfrak{p})$ consisting of linearly independent functions that, together with the constant function, can linearly express the components of any $U(\mathfrak{p})$ no greater than $n$-th order. 
  The minus log-determinant of the covariance matrix of $\bm{w}'_n(\mathfrak{p})$ is a constant different from $q_{2n}^{(0)}(\avrg{W}^k_j)$. 
\end{proposition}
\begin{proof}
The condition implies that the vector $\bm{w}'_n$ can be expressed as a linear transforamtion of $\bm{w}_n$: 
$$
\bm{w}'_n=A\bm{w}_n+\bm{b}, 
$$
where $A$ is a non-singular constant matrix, and $\bm{b}$ is a constant vector.
So we have
\begin{align*}
\bm{w}'_n\bm{w}'_n{}^t=&(A\bm{w}_n+\bm{b})(A\bm{w}_n+\bm{b})^t\\
=&A(\sum_jY_jw_{2n,j}+Y_0)A^t+A\bm{w}_n\bm{b}^t+\bm{b}\bm{w}_n^tA^t+\bm{b}\bm{b}^t. 
\end{align*}
The covariance matrix of $\bm{w}'_n$ is then given by 
\begin{align*}
  &A(\sum_jY_j\avrg{w}_{2n,j}+Y_0)A^t+A\avrg{\bm{w}}_n\bm{b}^t+\bm{b}\avrg{\bm{w}}_n^tA^t+\bm{b}\bm{b}^t
  -(A\avrg{\bm{w}}_n+\bm{b})(A\avrg{\bm{w}}_n+\bm{b})^t\\
  =&A(\sum_jY_j\avrg{w}_{2n,j}+Y_0-\avrg{\bm{w}}_n\avrg{\bm{w}}_n^t)A^t. 
\end{align*}
We could follow the same procedure for $\bm{w}_n$ to define a quasi-entropy as 
the minus log-determinant of the covariance matrix, which can be calculated as 
\begin{align*}
  &-\ln\det \Big(A(\sum_jY_j\avrg{w}_{2n,j}+Y_0-\avrg{\bm{w}}_n\avrg{\bm{w}}_n^t)A^t\Big)\\
  =&-\ln\det (\sum_jY_j\avrg{w}_{2n,j}+Y_0-\avrg{\bm{w}}_n\avrg{\bm{w}}_n^t)-\ln\det (A^tA)\\
  =&q_{2n}^{(0)}(\avrg{W}^k_j)-\ln\det(A^tA). 
\end{align*}
The constant difference is $-\ln\det(A^tA)$ depending only on the transformation matrix between $\bm{w}_n$ and $\bm{w}'_n$, which is independent of $\avrg{W}^k_j$. 
\end{proof}
Notice that distinction by a constant makes no difference in the free energy. 
Therefore, when discussing the quasi-entropy, we could choose the $\bm{w}(\mathfrak{p})$ according to the convenience of our discussion. 

In what follows, we prove that the quasi-entropy given in \eqref{logdet_general} satisfies the properties of the entropy by constrained minimization given in the previous section. 
We begin with a few technical results. 
\begin{lemma}\label{Ylind}
  The matrices $Y_j$ where $j=0,\ldots,d(2n)$ in \eqref{covmat} are linearly independent. 
\end{lemma}
\begin{proof}
  For our discussion below, we need to slightly refine Corollary \ref{w2n}. We shall show that Corollary \ref{w2n} still holds if we change the condition $0\le k,k'\le n$ into $1\le k,k'\le n$.
  In other words, we need to show that the components of $\bm{w}_n(\mathfrak{p})\big(\bm{w}_n(\mathfrak{p})\big)^t$ can linearly express the components of $\bm{w}_{2n}(\mathfrak{p})$ and the constant function. 
  Use induction. The case $n=1$ could be verified directly. 
  If $n>1$, Corollary \ref{w2n} actually indicates that the components of the matrix
  $$
  \left(\begin{array}{c}1\\\bm{w}_{n}(\mathfrak{p})\end{array}\right)\Big(1,\big(\bm{w}_n(\mathfrak{p})\big)^t\Big)=\left(
  \begin{array}{cc}
    1 & \bm{w}_n(\mathfrak{p})^t\\
    \bm{w}_n(\mathfrak{p}) & \bm{w}_n(\mathfrak{p})\big(\bm{w}_n(\mathfrak{p})\big)^t
  \end{array}
  \right) 
  $$
  are linearly equivalent to the components of $\bm{w}_{2n}(\mathfrak{p})$ together with the constant function. 
  However, the first row and the first column can be expressed linearly by $\bm{w}_{n_1}(\mathfrak{p})\big(\bm{w}_{n_1}(\mathfrak{p})\big)^t$ where $2n_1\ge n$, which completes the induction since $n_1<n$ if $n>1$. 
  
  Now we are ready to deal with the linear independence of $Y_j$. 
  Rearrange the components of the matrix $Y_j$ into a vector $\bm{y}_j$, and align them as $(\bm{y}_0,\ldots,\bm{y}_{d(2n)})$ to form a new matrix of the size $d(n)^2\times \big(d(2n)+1\big)$.
  This matrix is the coefficient matrix of the functions given by $\bm{w}_n(\mathfrak{p})\big(\bm{w}_n(\mathfrak{p})\big)^t$ under the basis $\bm{w}_{2n}(\mathfrak{p})$ and $1$. 
  The result from the previous paragraph indicates that its row rank equals to $d(2n)+1$. Therefore, it is column full-rank, which is exactly the linearly independence of $Y_j$. 
\end{proof}

\begin{lemma}\label{minuspd}
  Let $V$ be a $k\times k$ symmetric positive definite matrix, and $R$ be a symmetric positive semi-definite matrix of the same size.
  Assume that $V-R$ is also positive definite. Then we have 
  \begin{align*}
    \det(V-R)\le \det V. 
  \end{align*}
  The equality holds only when $R=0$. 
\end{lemma}
\begin{proof}
  We first deal with the case $R=\bm{v}\bm{v}^t$ where $\bm{v}$ is a vector with $k$ components. 
  Suppose $V-\bm{v}\bm{v}^t$ is also positive definite. Then we have 
  \begin{equation}
    \det (V-\bm{v}\bm{v}^t)\le \det V. \label{Vv}
  \end{equation}
  We can choose an orthogonal matrix $T$ such that only the first component of $T\bm{v}$ is nonzero, which we denote as $a$. 
  Denote $\tilde{V}=TVT^t$, and its $(1,1)$ cofactor by $M_{11}>0$.
  Denote by $E_{11}$ the matrix with the $(1,1)$ component equal to one and other components equal to zero. 
  Then we have 
  \begin{equation}
    \det(V-\bm{v}\bm{v}^t)=\det(\tilde{V}-aE_{11})=\det \tilde{V}-aM_{11}\le \det Q.\nonumber
  \end{equation}
  Generally, the positive semi-definite matrix $R$ can be expressed as 
  \begin{align*}
    R=\sum_{i=1}^k \bm{v}_i\bm{v}_i^t. 
  \end{align*}
  Use \eqref{Vv} for $k$ times to conclude the proof. 
\end{proof}

\begin{corollary}\label{diagest}
  Assume that $V$ is positive definite and can be written as 
  \begin{align*}
    V=\left(
    \begin{array}{cc}
      V_1 & R\\
      R^t & V_2
    \end{array}
    \right), 
  \end{align*}
  where $V_1$ and $V_2$ are square matrices.
  Then we have 
  \begin{equation}
    \det V\le \det V_1\det V_2. 
  \end{equation}
  The equality holds only when $R=0$. 
\end{corollary}
\begin{proof}
  Use block Gauss elimination,  
  \begin{align*}
    V=\left(
    \begin{array}{cc}
      I & 0\\
      RV_1^{-1} & I
    \end{array}
    \right)
    \left(
    \begin{array}{cc}
      V_1 & R^t\\
      0 & V_2-RV_1^{-1}R^t
    \end{array}
    \right). 
  \end{align*}
  Since both $V_1$ and $V_2$ are positive definite, using Lemma \ref{minuspd}, we deduce that 
  \begin{align*}
    \det V=\det V_1\det(V_2-RV_1R^t)\le \det V_1\det V_2. 
  \end{align*}
  The equality holds only when $RV_1R^t=0$. We could decompose $V_1=LL^t$ where $L$ is nonsingular lower-triangular. It means that $RL(RL)^t=0$, giving $RL=0$, then $R=0$. 
\end{proof}

\begin{lemma}\label{convexlogdet}
  The function $-\ln\det V$ is strictly convex about $V$ in the region of symmetric positive definite matrices. 
\end{lemma}
\begin{proof}
Assume that $V$ and $V\pm R$ are all $k\times k$ symmetric positive definite matrices. Then we can write $V=LL^t$, where $L$ is nonsingular. 
Thus, 
\begin{align*}
  \ln\det (V+R)=&\ln\det \big(L(I+L^{-1}R(L^t)^{-1})L^t\big)
  =\ln\det V + \ln\det (I+L^{-1}R(L^t)^{-1})\\
  =&\ln\det V + \sum_{i=1}^{k}\ln(1+\lambda_i), 
\end{align*}
where $\lambda_i$ are the eigenvalues of $L^{-1}R(L^t)^{-1}$.
Similarly, we have
\begin{align*}
  \ln\det (V-R)=\ln\det V + \sum_{i=1}^{k}\ln(1-\lambda_i). 
\end{align*}
Therefore, 
\begin{equation}
  \ln\det (V+R)+\ln\det (V+R)-2\ln\det V=\sum_{i=1}^{k}\ln(1-\lambda_i^2)\le 0. \nonumber
\end{equation}
The equality holds only when $\lambda_i=0$ for all $i$, which implies $R=0$. 
\end{proof}

We are now ready for the properties of $q_{2n}^{(0)}$. 
\begin{theorem}\label{q0prop}
  The quasi-entropy $q_{2n}^{(0)}(\avrg{W}^k_j)$ satisfies the following properties. 
  \begin{enumerate}
  \item The domain of $q_{2n}^{(0)}(\avrg{W}^k_j)$ is convex, and it is strictly convex on the domain. 
  \item It is invariant under rotations of two kinds: 
    \begin{enumerate}
    \item If $\avrg{W}^k_j$ is replaced by $\avrg{W}^k_j(\mathfrak{t})$; 
    \item If $\avrg{W}^k_j$ is replaced by $\sum_{j'=1}^{2k+1}\Gamma^k_{jj'}(\mathfrak{s})\avrg{W}^k_{j'}$. 
    \end{enumerate}
  \end{enumerate}
\end{theorem}
\begin{proof}
Choose two different values of $\avrg{\bm{w}}_{2n}$, which we denote as $\avrg{\bm{w}}_{[1],2n}, \avrg{\bm{w}}_{[2],2n}$. 
For the matrix $C(\avrg{\bm{w}}_{2n})$, we have 
\begin{align*}
  \frac{1}{2}\Big(C(\avrg{\bm{w}}_{[1],2n})&+C(\avrg{\bm{w}}_{[2],2n})\Big)\\
  =&\sum_{j=1}^{d(2n)} Y_j\frac{1}{2}(\avrg{w}_{[1],2n,j}+\avrg{w}_{[2],2n,j})+Y_0-\frac{1}{4}(\avrg{\bm{w}}_{[1],n}+\avrg{\bm{w}}_{[2],n})(\avrg{\bm{w}}_{[1],n}+\avrg{\bm{w}}_{[2],n})^t\\
  &-\frac{1}{4}(\avrg{\bm{w}}_{[1],n}-\avrg{\bm{w}}_{[2],n})(\avrg{\bm{w}}_{[1],n}-\avrg{\bm{w}}_{[2],n})^t\\
  =&C\big(\frac{1}{2}(\avrg{\bm{w}}_{[1],2n}+\avrg{\bm{w}}_{[2],2n})\big)-\frac{1}{4}(\avrg{\bm{w}}_{[1],n}-\avrg{\bm{w}}_{[2],n})(\avrg{\bm{w}}_{[1],n}-\avrg{\bm{w}}_{[2],n})^t. 
\end{align*}
Therefore, if the two matrices $C(\avrg{\bm{w}}_{[1],2n})$ and $C(\avrg{\bm{w}}_{[2],2n})$ are both positive definite, the matrix $C\big(\frac{1}{2}(\avrg{\bm{w}}_{[1],2n}+\avrg{\bm{w}}_{[2],2n})\big)$ is also positive definite. 
Using Lemma \ref{minuspd} and \ref{convexlogdet}, we deduce that 
\begin{align*}
  &\frac{1}{2}\Big(-\ln\det C(\avrg{\bm{w}}_{[1],2n})-\ln\det C(\avrg{\bm{w}}_{[2],2n})\Big)\\
  \ge & -\ln\det \Big(C\big(\frac{1}{2}(\avrg{\bm{w}}_{[1],2n}+\avrg{\bm{w}}_{[2],2n})\big)-\frac{1}{4}(\avrg{\bm{w}}_{[1],n}-\avrg{\bm{w}}_{[2],n})(\avrg{\bm{w}}_{[1],n}-\avrg{\bm{w}}_{[2],n})^t\Big)\\
  \ge & -\ln\det C\big(\frac{1}{2}(\avrg{\bm{w}}_{[1],2n}+\avrg{\bm{w}}_{[2],2n})\big). 
\end{align*}
By noticing Lemma \ref{Ylind}, the equality holds only if $\avrg{\bm{w}}_{[1],2n}=\avrg{\bm{w}}_{[2],2n}$. 

For the invariance under rotations, 
we go back to the construction of covariance matrix starting from \eqref{covorg}. 
Let us reformulate the conditions (a) and (b) below, to find out what averaged tensors equal to $\avrg{W}^k_j$ after the rotations. 
\begin{enumerate}[(a)]
\item From $\langle W^k_{j}(\mathfrak{p})\rangle=\avrg{W}^k_j(\mathfrak{t})$, we deduce that 
$$
{W}^k_i\cdot\langle W^k_{j}(\mathfrak{t}^{-1}\mathfrak{p})\rangle={W}^k_i(\mathfrak{t})\cdot\langle W^k_{j}(\mathfrak{p})\rangle={W}^k_i(\mathfrak{t})\cdot\avrg{W}^k_j(\mathfrak{t})={W}^k_i\cdot\avrg{W}^k_j. 
$$
\item From $\langle W^k_j(\mathfrak{p})\rangle=\sum_{j'=1}^{2k+1}\Gamma^k_{jj'}(\mathfrak{s})\avrg{W}^k_{j'}$, we use \eqref{basistrans} and \eqref{rotmat} to deduce that 
\begin{align*}
  \langle W^k_{j}(\mathfrak{ps}^{-1})\rangle=&\sum_{j'}\Gamma_{jj'}(\mathfrak{s}^{-1})\langle W^k_{j'}(\mathfrak{p})\rangle
  \\
  =&\sum_{j',j''}\Gamma_{jj'}(\mathfrak{s}^{-1})\Gamma_{j'j''}(\mathfrak{s})\avrg{W}^k_{j''}\\
  =&\sum_{j''}\delta_{jj''}\avrg{W}^k_{j''}\\
  =&\avrg{W}^k_{j}. 
\end{align*}
So we have 
$$
{W}^k_i\cdot\langle W^k_{j}(\mathfrak{ps}^{-1})\rangle={W}^k_i\cdot\avrg{W}^k_j.
$$
\end{enumerate}
For the quasi-entropy, let us use a different $\bm{w}'_n$ for the covariance matrix \eqref{covmat}: 
in the case (a), we let $\bm{w}'_n$ consist of $W^k_i\cdot W^k_j(\mathfrak{t}^{-1}\mathfrak{p})$; in the case (b), we let $\bm{w}'_n$ consist of $W^k_i\cdot W^k_j(\mathfrak{ps}^{-1})$. 
In both cases, we shall notice that $Y_j$ in \eqref{covmat} are identical for $\bm{w}_n$.
Since the averaged values are also identical, by Proposition \ref{anotherw} we only need to check that the constant difference between two definitions is zero. 
For this purpose, we examine the transformation matrix $A$. 
By \eqref{introt}, we have
\begin{align*}
  &\int\big(W^k_i\cdot W^k_j(\mathfrak{t}^{-1}\mathfrak{p})\big)\big(W^{k'}_{i'}\cdot W^{k'}_{j'}(\mathfrak{t}^{-1}\mathfrak{p})\big)\,\md\mathfrak{p}\\
  =&\int\big(W^k_i\cdot W^k_j(\mathfrak{ps}^{-1})\big)\big(W^{k'}_{i'}\cdot W^{k'}_{j'}(\mathfrak{ps}^{-1})\big)\,\md\mathfrak{p}\\
  =&\int\big(W^k_i\cdot W^k_j(\mathfrak{p})\big)\big(W^{k'}_{i'}\cdot W^{k'}_{j'}(\mathfrak{p})\big)\,\md\mathfrak{p}. 
\end{align*}
Hence, we deduce for both cases that 
$$
\int \bm{w}_n(\bm{w}_n)^t\,\md\mathfrak{p}=\int \bm{w}'_n(\bm{w}'_n)^t\,\md\mathfrak{p}=A\left(\int \bm{w}_n(\bm{w}_n)^t\,\md\mathfrak{p}\right)A^t. 
$$
Calculating the determinant, we arrive at $\det(A^tA)=1$. 
The constant difference is then given by $-\ln \det(A^tA)=0$. 

\end{proof}

The function $q_{2n}^{(0)}$ is a function of all the symmetric traceless tensors up to the $2n$-th order.
We could define the quasi-entropy for any choice of symmetric traceless tensors using this function. 
Assume that some tensors $U_j$ of the order $n_j$ are chosen, with their averaged values $\langle U_j(\mathfrak{p})\rangle=\avrg{U}_j$ given. 
In the last paragraph of the previous section, we state the entropy for $\avrg{U}_j$ as a constrained minimization of the entropy $f_m$ about all $\avrg{W}^k_j$ for $k\le m$, where we require $m\ge n_j$. 
Here, we substitute the function $f_m$ with $q_{2n}^{(0)}$ where $2n\ge n_j$, without changing the other settings. 
In other words, we define the quasi-entropy $q_{2n}(\avrg{U}_j)$ about $\avrg{U}_j$ as the solution to the following constrained minimization problem, 
\begin{align}
  q_{2n}(\avrg{U}_j)=\min_{\avrg{W}^k_j} q_{2n}^{(0)}(\avrg{W}^k_j),\quad \mbox{s.t. }\sum_{j'=1}^{2n_j+1} u_{jj'}\avrg{W}^{n_j}_{j'}=\avrg{U}_j. \label{minprob2}
\end{align}
Here, the coefficients $u_{jj'}$ are determined from \eqref{coefU}. 
It shall be noted that the constraints are all linear. 

\begin{proposition}\label{minexist}
  The minimization problem in \eqref{minprob2} possesses a unique solution, if the feasible set, in which the constraints hold and the matrix $C(\avrg{\bm{w}}_{2n})$ is positive definite, is nonempty. 
\end{proposition}
\begin{proof}
  Since the constraints are linear, the uniqueness is guaranteed by the strict convexity. In the following, we show the existence. 
  
  First, in the matrix $C(\avrg{\bm{w}}_{2n})$, we show that the diagonal elements of $\sum_jY_j\avrg{w}_{2n,j}+Y_0$ are bounded by some constant.
  Since $\sum_jY_j\avrg{w}_{2n,j}+Y_0$ is positive definite, every diagonal element must be positive.
  Meanwhile, $W^k_j$ is an orthogonal basis of $k$-th order symmetric traceless tensors. So we have 
  $$
  \sum_{i=1}^{2k+1}\frac{1}{|W^k_i|^2}\big(W^k_i\cdot W^k_j(\mathfrak{p})\big)^2=|W^k_j|^2, 
  $$
  which is a constant. 
  Thus, the diagonal elements of $\sum_jY_j\avrg{w}_{2n,j}+Y_0$ can be divided into several subsets.
  For the elements in each subset, a linear combination with positive coefficients is a constant, which is sufficient for the upper-boundedness. 

  As a result, the off-diagonal elements of $\sum_jY_j\avrg{w}_{2n,j}+Y_0$ are also bounded by the constant. This can be realized by noticing that any $2\times 2$ principal minor is positive definite. 
  Besides, based on Corollary \ref{diagest}, we deduce from the upper-boundedness of diagonal elements that $-\ln\det C(\avrg{\bm{w}}_{2n})$ is bounded from below. 
  
  Next, we claim that the domain of the $q_{2n}^{(0)}$ is bounded for $\avrg{\bm{w}}_{2n}$. By the derivation of the previous two paragraphs, every element in the matrix $\sum_jY_j\avrg{w}_{2n,j}+Y_0$ yields a linear inequality about $\avrg{\bm{w}}_{2n}$. Let us put them together, so that they can be expressed in the form 
  $$
  \Big|\sum_{j}B_{ij}\avrg{w}_{2n,j}+a_i\Big|\le b_i, \quad i=1,\ldots,d(n)^2. 
  $$
  where $B$ is some matrix and $a_i,b_i$ are constants. 
  Now, by Lemma \ref{Ylind}, the matrix $B$ is column full-rank. 
  This implies that the region given by the inequalities is bounded. 
  
  Now, let us choose a sequence $\avrg{\bm{w}}_{[l],2n}$ such that 
  \begin{align}
    \lim_{l\to +\infty}-\ln\det C(\avrg{\bm{w}}_{[l],2n})=\inf_{\sum u_{jj'}\avrg{W}^{n_j}_{j'}=\avrg{U}_j}q_{2n}^{(0)}(\avrg{W}^k_j). \label{seqinf}
  \end{align}
  Since $\avrg{\bm{w}}_{[l],2n}$ is bounded, we can choose a subsequence such that $\avrg{\bm{w}}_{[l],2n}$ converges to a $\avrg{\bm{w}}_{[0],2n}$. The linear constraints still hold for $\avrg{\bm{w}}_{[0],2n}$. If $C(\avrg{\bm{w}}_{[0],2n})$ is positive definite, it gives a minimizer. We shall eliminate the possibility of $C(\avrg{\bm{w}}_{[0],2n})$ being a singular positive semi-definite matrix.
  If it is the case, we can deduce that 
  \begin{align*}
    \lim_{l\to +\infty}\det C(\avrg{\bm{w}}_{[l],2n})=\det C(\avrg{\bm{w}}_{[0],2n})=0, 
  \end{align*}
  which contradicts \eqref{seqinf}. 
\end{proof}

\begin{theorem}\label{qprop}
The quasi-entropy \eqref{minprob2} satisfies three properties. 
\begin{enumerate}
\item Its domain is convex, in which the function is strict convex about $\avrg{U}_j$. 
\item Invariance under rotations. 
  \begin{enumerate}
  \item If we replace $\avrg{U}_j$ with $\avrg{U}_j(\mathfrak{t})$, the minimum value is the same. 
  \item We change the constraints into 
    \begin{equation}
      \sum_{j',j''=1}^{2n_j+1} u_{jj'}\Gamma_{j'j''}(\mathfrak{s})\avrg{W}^{n_j}_{j''}=\avrg{U}_j. 
    \end{equation}
  \end{enumerate}
\item Assume that for all $j$, we have $\sum_{j'}u_{jj'}W^{n_j}_{j'}(\mathfrak{p})\in \mathbb{A}^{\mathcal{G},n_j}$. For any $l$-th order symmetric traceless tensor $U(\mathfrak{p})\in(\mathbb{A}^{\mathcal{G},l})^{\perp}$ for $l\le m$, let us express it by $W^l_j$:
  $$
  U(\mathfrak{p})=\sum_{j=1}^{2l+1}u_jW^l_j(\mathfrak{p}). 
  $$
  The minimizer must satisfy $\sum_{j=1}u_j\avrg{W}^l_j=0$. 
\end{enumerate}
\end{theorem}
\begin{proof}
  The domain of $q_{2n}(\avrg{U}_j)$ is given by: there exists $\avrg{W}^k_j$, such that $\sum_{j'}u_{jj'}W^{n_j}_{j'}=\avrg{U}_j$ and $C(\avrg{\bm{w}}_{2n})$ is positive definite.
  In this case, there exists a unique minimizer $\avrg{W}^k_j$ of \eqref{minprob2} according to Proposition \ref{minexist}. 
  Using similar arguments in Theorem \ref{q0prop} for the minimizer $\avrg{W}^k_j$, the first two properties can be established. 



  We focus on the last property. 
  %
  For our discussion afterwards, we choose a particular complex-valued basis 
  \begin{align}
    W^k_j(\mathfrak{p})=\left\{
    \begin{array}{ll}
      p_{k,j}(\bm{m}_1,\mathfrak{i})(\bm{m}_2+\sqrt{-1}\bm{m}_3)^{j-1}, & 1\le j\le k+1;\\
      p_{k,j}(\bm{m}_1,\mathfrak{i})(\bm{m}_2-\sqrt{-1}\bm{m}_3)^{j-k-1}, &k+2\le j\le 2k+1. 
    \end{array}
    \right.\label{parbasis}
  \end{align}
  Here, we recall that a polynomial of $\bm{m}_i$ gives a symmetric tensor.
  The polynomial $p_{k,j}$ is defined from the $(k-j)$-th order Jacobi polynomial with the indices $(j,j)$ \cite{xu_tensors}, which we do not write down explicitly because it is unrelated to our discussion. 
  Using the Euler angles, we could write 
  \begin{align*}
    \bm{m}_2\pm\sqrt{-1}\bm{m}_3=\left(\left(
    \begin{array}{c}
      -\sin\alpha\\
      \cos\alpha\cos\beta\\
      \cos\alpha\sin\beta
    \end{array}
    \right)
    +\sqrt{-1}
    \left(
    \begin{array}{c}
      0\\
      -\sin\beta\\
      \cos\beta
    \end{array}
    \right)\right)
    e^{\mp\sqrt{-1}\gamma}.
  \end{align*}
  Meanwhile, $\bm{m}_1$ only depends on $\alpha$ and $\beta$.
  Therefore, the basis tensors can be written in the form 
  \begin{equation}
    p_{k,j}(\bm{m}_1,\mathfrak{i})(\bm{m}_2\pm\sqrt{-1}\bm{m}_3)^{j}=V^k(\alpha,\beta)e^{\mp\sqrt{-1}j\gamma}, \label{tensor_basis}
  \end{equation}
  where $V^k(\alpha,\beta)$ is a $k$-th order symmetric traceless tensor independent of $\gamma$. 
  Any of its component can be written in the separate variable form $v(\alpha,\beta)e^{\mp\sqrt{-1}j\gamma}$. 
  When constructing the covariance matrix, since the basis is complex-valued, we need to replace all the transpose $\bm{w}_n^t$ in the preceding by conjugate transpose $\bm{w}_n^*$. 

  We first discuss the special case where the group $\mathcal{G}$ is $\mathcal{C}_m=\{\mathfrak{i},\mathfrak{j}_{2\pi/m},\ldots,\mathfrak{j}_{2(n-1)\pi/m}\}$ generated by a rotation by the angle $2\pi/m$.
  Let us look into the invariant tensors of the group. 
  Suppose that the Euler angles of $\mathfrak{p}$ are $(\alpha,\beta,\gamma)$. 
  Then, the Euler angles of $\mathfrak{pj}_{2\pi/m}$ are $(\alpha,\beta,\gamma+2\pi/m)$ (recall the discussion below \eqref{jtheta}). 
  Therefore, for every tensor in $\mathbb{A}^{\mathcal{C}_m,l}$, its components can only have the terms such that $j$ is a multiple of $m$. 
  Correspondingly, for every tensor in $(\mathbb{A}^{\mathcal{C}_m,l})^{\perp}$, its components can only have the terms such that $j$ is not a multiple of $m$. 





  Let us rearrange the components of the vector $\bm{w}_n$ (for any $n$) according to the $j$ in $e^{\sqrt{-1}j\gamma}$ by $j \bmod m$, such that 
  $$
  \bm{w}_{n}=\left(
  \begin{array}{c}
    \bm{w}_n^{(0)}\\
    \bm{w}_n^{(1)}\\
    \vdots\\
    \bm{w}_n^{(m-1)}
  \end{array}
  \right), 
  $$
  where $\bm{w}_n^{(i)}$ consists of all the functions with $j\equiv i\,\pmod m$. 
  The condition of the theorem indicates that in $\avrg{\bm{w}}_{2n}$, only the values of $\avrg{\bm{w}}_{2n}^{(0)}$ are given as constraints in the minimization problem \eqref{minprob2}. 
  In the covariance matrix $C(\avrg{\bm{w}}_{2n})$, by checking the $j$ in the factor $e^{\sqrt{-1}j\gamma}$, we recognize that the components of $\avrg{\bm{w}}_{2n}^{(0)}$ are all located in the diagonal blocks, actually in $\langle\bm{w}_n^{(0)}(\mathfrak{p})\bm{w}_n^{(0)}(\mathfrak{p})^*\rangle-\langle\bm{w}_n^{(0)}(\mathfrak{p})\rangle\langle\bm{w}_n^{(0)}(\mathfrak{p})\rangle^*$ and $\langle\bm{w}_n^{(i)}(\mathfrak{p})\bm{w}_n^{(i)}(\mathfrak{p})^*\rangle$. 
  The minimization in \eqref{minprob2} is about all the other components, $\avrg{\bm{w}}_{2n}^{(i)}$ for $1\le i\le m-1$. 
  For the determinant, Corollary \ref{diagest} gives the estimate 
  \begin{align*}
    \det&\Big(\langle\bm{w}_n(\mathfrak{p})\bm{w}_n(\mathfrak{p})^*\rangle-\langle\bm{w}_n(\mathfrak{p})\rangle\langle\bm{w}_n(\mathfrak{p})\rangle^*\Big)\\
    &\le \prod_{j=0}^{m-1}\det\Big(\langle\bm{w}_n^{(j)}(\mathfrak{p})\bm{w}_n^{(j)}(\mathfrak{p})^*\rangle-\langle\bm{w}_n^{(j)}(\mathfrak{p})\rangle\langle\bm{w}_n^{(j)}(\mathfrak{p})\rangle^*\Big)\\
    &\le \det\Big(\langle\bm{w}_n^{(0)}(\mathfrak{p})\bm{w}_n^{(0)}(\mathfrak{p})^*\rangle-\langle\bm{w}_n^{(0)}(\mathfrak{p})\rangle\langle\bm{w}_n^{(0)}(\mathfrak{p})\rangle^*\Big)\prod_{j=1}^{m-1}\det\langle\bm{w}_n^{(j)}(\mathfrak{p})\bm{w}_n^{(j)}(\mathfrak{p})^*\rangle. 
  \end{align*}
  The equality can indeed be attained, by letting $\avrg{\bm{w}}_n^{(j)}=0$ for $j\ne 0$.
  Note that $\avrg{\bm{w}}_n^{(j)}$ gives all the linearly independent components of the vanishing tensors of $\mathcal{C}_m$. Therefore, we prove the third property for the group $\mathcal{C}_m$. 

  For any cyclic group $\mathcal{G}$, we could choose an $\mathfrak{s}\in SO(3)$ such that $\mathcal{G}=\mathfrak{s}\mathcal{C}_m\mathfrak{s}^{-1}$ for some $m$.
  We verify that $U(\mathfrak{ps})\in \mathbb{A}^{\mathcal{G},l}$ is equivalent to $U(\mathfrak{p})\in \mathbb{A}^{\mathcal{C}_m,l}$. Let $U_1(\mathfrak{p})=U(\mathfrak{ps})$ and $\mathfrak{g}=\mathfrak{sj}_{\theta}\mathfrak{s}^{-1}$ where $\mathfrak{j}_{\theta}\in\mathcal{C}_m$. Then we have 
  \begin{align*}
    U_1(\mathfrak{pg})=U(\mathfrak{pgs})=U(\mathfrak{psj}_{\theta})=U(\mathfrak{ps})=U_1(\mathfrak{p}). 
  \end{align*}
  Hence, $U(\mathfrak{ps})\in (\mathbb{A}^{\mathcal{G},l})^{\perp}$ is equivalent to $U(\mathfrak{p})\in (\mathbb{A}^{\mathcal{C}_m,l})^{\perp}$. 
  Recall that the function $q_{2n}^{(0)}$ is invariant under rotations.
  Hence, we could substitute $W^k_j(\mathfrak{p})$ with $W^k_j(\mathfrak{ps})$ in \eqref{parbasis} and repeat the derivation in the previous paragraph. 
  In this way, we prove the third property for any cyclic group $\mathcal{G}$. 

  For a general point group $\mathcal{G}\subseteq SO(3)$, the above discussion is suitable for any of its cyclic subgroup. Every rotation in $\mathcal{G}$ can generate a cyclic subgroup of $\mathcal{G}$. Therefore, the third property is obtained immediately by noticing  
  \begin{align*}
    \mathbb{A}^{\mathcal{G},l}=\bigcap_{\substack{\mathcal{G}_1\subseteq\mathcal{G}\\\mathcal{G}_1\text{ cyclic}}} \mathbb{A}^{\mathcal{G}_1,l},\qquad
    (\mathbb{A}^{\mathcal{G},l})^{\perp}=\sum_{\substack{\mathcal{G}_1\subseteq\mathcal{G}\\\mathcal{G}_1\text{ cyclic}}} (\mathbb{A}^{\mathcal{G}_1,l})^{\perp}. 
  \end{align*}
\end{proof}



Theorem \ref{qprop}, especially the third property, indicates that although the density is not involved in the quasi-entropy, the symmetry arguments can still be utilized. This will bring great convenience when constructing the quasi-entropy for particular molecular symmetry and choice of tensors. 

\section{Quasi-entropy for different molecular symmetries\label{qentexpl}}
In this section, we use the properties of the quasi-entropy presented in the previous section to write down the quasi-entropy for several molecular symmetries. 
As we have mentioned earlier, for certain molecular symmetry, all the rotations allowed by the symmetry form a subgroup $\mathcal{G}$ of $SO(3)$.
Then, some symmetric traceless tensors $U_j$ invariant under the rotations in $\mathcal{G}$ are chosen, so that the averaged tensors $\langle U_j(\mathfrak{p})\rangle=\avrg{U}_j$ are nonvanishing and act as the order parameters. 
To construct the quasi-entropy for these averaged tensors, we choose the minimum $n$, such that $2n$ is no less than the highest order of $U_j$.
Then, we write down the function $q_{2n}(\avrg{U}_j)$ given by \eqref{minprob2}, 
where we utilize the third property in Theorem \ref{qprop} that vanishing tensors of $\mathcal{G}$ are zero to simplify the expression. 
We will write down explicit expressions for axisymmetries and two-fold symmetries based on $q_2^{(0)}$, 
for tetrahdedral, octahedral, three-fold and four-fold symmetries based on $q_4^{(0)}$. 

As pointed out in the proof of Theorem \ref{qprop}, for the group $\mathfrak{s}\mathcal{G}\mathfrak{s}^{-1}$, the invariant tensors are $U(\mathfrak{ps})$ where $U(\mathfrak{p})\in\mathbb{A}^{\mathcal{G},l}$.
A different choice of $\mathfrak{s}$ can be interpreted by a different posing of the body-fixed frame on a rigid molecule. 
Thus, we only need to choose a specific $\mathfrak{s}$ (a specific posing of the body-fixed frame) that simplifies our presentation.
We shall explain by rod-like molecules, whose symmetry is described by the group $\mathcal{D}_{\infty}$ (a Sch\" oflies notation). 
If we choose the body-fixed frame appropriately such that $\bm{m}_1$ is the rotation axis, the group $\mathcal{D}_{\infty}$ contains $\mathfrak{j}_{\theta}$ in \eqref{jtheta} for any $\theta$. 
It also allows the rotations that bring $\bm{m}_1$ to $-\bm{m}_1$, such as $\mathrm{diag}(-1,1,-1)$ that transforms $(\bm{m}_1,\bm{m}_2,\bm{m}_3)$ into $(-\bm{m}_1,\bm{m}_2,-\bm{m}_3)$. 
In this case, up to second order, the only invariant tensor is $\bm{m}_1^2-\mathfrak{i}/3$. 
Below, we do not explain the choice of body-fixed frame. 
Inseatd, we will write down directly the groups using the Sch\" oflies notations, together with the nonvanishing tensors following the notations in \cite{xu_tensors}. 

We mention that a rigid molecule may also be invariant under improper rotations. 
But they do not affect the invariant tensors, thus are not reflected in the quasi-entropy. 
The improper rotations play a role the interaction terms in the free energy, which we refer to another work for interested readers \cite{xu_kernel}. 

\subsection{Quasi-entropy based on $q^{(0)}_2$}
We write down the function $q_2^{(0)}$ that is constructed based on the vector $\bm{w}_1$. 
Choose $W^1_j(\mathfrak{p})=\bm{m}_j$, so that $\bm{w}_1$ is a nine-component vector consisting of $W^1_i\cdot W^1_j(\mathfrak{p})=\bm{e}_i\cdot\bm{m}_j$.
It can be written as 
$$
\bm{w}_1=\left(
\begin{array}{c}
  \bm{m}_1\\
  \bm{m}_2\\
  \bm{m}_3
\end{array}
\right). 
$$
The function $q_2^{(0)}$ is given by 
\begin{align*}
  &q_2^{(0)}
  =-\ln\det(\langle\bm{w}_1\bm{w}_1^t\rangle-\langle\bm{w}_1\rangle\langle\bm{w}_1\rangle^t)\\
  &=-\ln\det\left(\left(
  \begin{array}{ccc}
    \langle \bm{m}_1^2\rangle & \langle \bm{m}_1\otimes\bm{m}_2\rangle & \langle \bm{m}_1\otimes\bm{m}_3\rangle\\
    \langle \bm{m}_2\otimes\bm{m}_1\rangle & \langle \bm{m}_2^2\rangle & \langle \bm{m}_2\otimes\bm{m}_3\rangle\\
    \langle \bm{m}_3\otimes\bm{m}_1\rangle & \langle \bm{m}_3\otimes\bm{m}_2\rangle & \langle \bm{m}_3^2\rangle
  \end{array}
  \right)
  -\left(
  \begin{array}{c}
    \langle\bm{m}_1\rangle\\
    \langle\bm{m}_2\rangle\\
    \langle\bm{m}_3\rangle
  \end{array}
  \right)(\langle\bm{m}_1\rangle^t,\langle\bm{m}_2\rangle^t,\langle\bm{m}_3\rangle^t)\right). 
\end{align*}
In the above, we actually regard a second order tensor equivalently as a $3\times 3$ matrix, with the $(i,j)$ component located in the $i$-th row, $j$-th column. 
Recall again the monomial notation for symmetric tensors in \eqref{tensor_monomial}, such as $\bm{m}_1\bm{m}_2=\frac{1}{2}(\bm{m}_1\otimes\bm{m}_2+\bm{m}_1\otimes\bm{m}_2)$. 
The function $q_2^{(0)}$ is a function of eight tensors, 
including three first order tensors $\langle \bm{m}_1\rangle$, 
$\langle \bm{m}_2\rangle$,
$\langle \bm{m}_3\rangle$,
and five second order tensors 
$\langle \bm{m}_1^2-\mathfrak{i}/3\rangle$,
$\frac{1}{2}\langle \bm{m}_2^2-\bm{m}_3^2\rangle$,
$\langle \bm{m}_1\bm{m}_2\rangle$,
$\langle \bm{m}_1\bm{m}_3\rangle$,
$\langle \bm{m}_2\bm{m}_3\rangle$.
To recognize this, we express the diagonal blocks as 
\begin{align*}
  \langle \bm{m}_1^2\rangle=&\langle \bm{m}_1^2-\frac{\mathfrak{i}}{3}\rangle+\frac{I_3}{3}\\
  \langle \bm{m}_2^2\rangle=&\frac{1}{2}(\langle \bm{m}_2^2\rangle+\langle \bm{m}_3^2\rangle)+\frac{1}{2}(\langle \bm{m}_2^2\rangle-\langle \bm{m}_3^2\rangle)\\
  =&\frac{1}{3}I_3-\frac{1}{2}\langle \bm{m}_1^2-\frac{\mathfrak{i}}{3}\rangle+\frac{1}{2}\langle \bm{m}_2^2-\bm{m}_3^2\rangle, \\
  \langle \bm{m}_2^2\rangle=&\frac{1}{3}I_3-\frac{1}{2}\langle \bm{m}_1^2-\frac{\mathfrak{i}}{3}\rangle-\frac{1}{2}\langle \bm{m}_2^2-\bm{m}_3^2\rangle, 
\end{align*}
where we use \eqref{squaresum} in the second equality, and $I_3$ represents the $3\times 3$ identity matrix\footnote{Actually $\mathfrak{i}$ and $I_3$ are two different notations for the $3\times 3$ identity matrix. We use $\mathfrak{i}$ when it comes with $\bm{m}_i$, and $I_3$ for averaged tensors.}. 
The off-diagonal blocks can be calculated as 
\begin{align*}
  \langle \bm{m}_1\otimes\bm{m}_2\rangle_{ij}
  &=\langle \bm{m}_1\bm{m}_2\rangle_{ij}+\frac{1}{2}(\langle \bm{m}_1\otimes\bm{m}_2\rangle_{ij}-\langle \bm{m}_2\otimes\bm{m}_1\rangle_{ij})\\
  &=\langle \bm{m}_1\bm{m}_2\rangle_{ij}+\frac{1}{2}\sum_{s=1}^3\epsilon_{ijs}\langle \bm{m}_3\rangle_s, 
\end{align*}
so that $\langle \bm{m}_1\otimes\bm{m}_2\rangle$ can be expressed by a second order symmetric traceless tensor and a first order tensor. 
The other blocks can be calculated similarly. 



\subsubsection{Rod-like molecules, group $\mathcal{D}_{\infty}$}
Up to the second order, the only nonvanishing tensor is $Q=\langle\bm{m}_1^2-\mathfrak{i}/3\rangle$ that is the commonly used one in the literature. 
To derive $q_2(Q)$, according to \eqref{minprob2}, $q_2^{(0)}$ is minimized with $Q$ fixed. 
By Theorem \ref{qprop}, the other seven averaged tensors shall vanish, 
$$
\langle\bm{m}_1\rangle=\langle\bm{m}_2\rangle=\langle\bm{m}_3\rangle=\langle\bm{m}_2^2-\bm{m}_3^2\rangle=\langle\bm{m}_1\bm{m}_3\rangle=\langle\bm{m}_1\bm{m}_2\rangle=\langle\bm{m}_2\bm{m}_3\rangle=0. 
$$
Setting these tensors zero in $q_2^{(0)}$, we obtain the resulting quasi-entropy, 
\begin{equation}
  q_2(Q)=-\ln\det (Q+\frac{I_3}{3})-2\ln\det (\frac{I_3}{3}-\frac{Q}{2}), \label{Dinf}
\end{equation}
which is defined for $Q$ symmetric traceless. 

We point out that this function also satisfies the asymptotics of $f_{\mathrm{ent}}(Q)$ given in \eqref{entmin0}. 
Using the invariance under rotations in Theorem \ref{qprop}, we assume that $Q$ is diagonal, i.e. $Q=\mbox{diag}(q_1,q_2,q_3)$ with $q_1+ q_2+ q_3=0$, otherwise we could substitute $Q$ with $Q(\mathfrak{t})$ for some $\mathfrak{t}$ for diagonalization. 
The function $q_2(Q)$ can now be written as 
\begin{equation}
  \sum_{i=1}^3 -\ln (q_i+1/3)-2\ln(1/3-q_i/2). 
\end{equation}
We can see that the function enforces $-1/3<q_i< 2/3$.
Suppose that $q_1\le q_2\le q_3$. 
When $q_1\to (-1/3)^+$, the leading order of the function is $-\ln(q_1+1/3)$. 
This is consistent with the results for $f_{\mathrm{ent}}(Q)$ in
\cite{ball2010nematic}. 

\subsubsection{Group $\mathcal{C}_{\infty}$}
Up to the second order, the nonvanishing tensors are 
$$
Q^1=\langle\bm{m}_1\rangle, Q^2=\langle\bm{m}_1^2-\frac{1}{3}\mathfrak{i}\rangle.
$$
The other averaged tensors are all zero: $\langle\bm{m}_2\rangle=\langle\bm{m}_3\rangle=\langle\bm{m}_2^2-\bm{m}_3^2\rangle=\langle\bm{m}_1\bm{m}_3\rangle=\langle\bm{m}_1\bm{m}_2\rangle=\langle\bm{m}_2\bm{m}_3\rangle=0$. 
Thus, we can calculate $\langle\bm{m}_2^2\rangle=\langle\bm{m}_3^2\rangle=I_3/3-Q^1/2$, $\langle\bm{m}_1\otimes\bm{m}_3\rangle=\langle\bm{m}_1\otimes\bm{m}_2\rangle=0$, and 
$$
\langle\bm{m}_2\otimes\bm{m}_3\rangle_{ij}=\frac{1}{2}\langle\bm{m}_2\bm{m}_3\rangle_{ij}+\frac{1}{2}\langle\bm{m}_2\otimes\bm{m}_3-\bm{m}_3\otimes\bm{m}_2\rangle_{ij}=\frac{1}{2}\sum_{s=1}^3\epsilon_{ijs}Q^1_{s}.
$$
Denote $S_{ij}=\sum_s\epsilon_{ijs}Q^1_{s}$. 
The quasi-entropy is then written as 
\begin{equation}
  q_2(Q^1,Q^2)=-\ln\det (\frac{I_3}{3}+Q^2-Q^1\otimes Q^1)-\ln\det\left(
  \begin{array}{cc}
    \frac{I_3}{3}-\frac{Q^2}{2} & \frac{S}{2}\\
    -\frac{S}{2} & \frac{I_3}{3}-\frac{Q^2}{2}
  \end{array}
  \right). \label{Cinf}
\end{equation}
It is a function about $Q^1$ and $Q^2$. 
In some cases, only $Q^1$ is kept as the order parameter.
If this is the case, we can define the quasi-entropy for $Q^1$ by minimizing the above function about $Q^2$, 
\begin{equation}
  q_2(Q^1)=\min_{Q^2}q_2(Q^1,Q^2). 
\end{equation}

\subsubsection{Group $\mathcal{C}_2$}
Up to the second order, the nonvanishing tensors are $Q^1$, $Q^2$, $M_1^2=\frac{1}{2}\langle\bm{m}_2^2-\bm{m}_3^2\rangle$, and $M_2^2=\langle\bm{m}_2\bm{m}_3\rangle$. 
Taking $\langle\bm{m}_2\rangle=\langle\bm{m}_3\rangle=\langle\bm{m}_1\bm{m}_3\rangle=\langle\bm{m}_1\bm{m}_2\rangle=0$ into $Q_2^{(0)}$, we arrive at the quasi-entropy  
\begin{equation}
  q_2(Q^1,Q^2,M_1^2,M_2^2)=-\ln\det (\frac{I_3}{3}+Q^2-Q^1Q^1)-\ln\det\left(
  \begin{array}{cc}
    \frac{I_3}{3}-\frac{Q^2}{2}+M_1^2 & M^2_2+\frac{S}{2}\\
    M_2^2-\frac{S}{2} & \frac{I_3}{3}-\frac{Q^2}{2}-M_1^2
  \end{array}
  \right). \label{C2}
\end{equation}
Similar to the group $\mathcal{C}_{\infty}$, we could construct quasi-entropy for part of the four tensors. For example, if a quasi-entropy is needed for $(Q^1,Q^2,M_1^2)$, we could define as the minimization 
\begin{equation}
  q_2(Q^1,Q^2,M_1^2)=\min_{M_2^2}q_2(Q^1,Q^2,M_1^2,M_2^2). \label{C2simp}
\end{equation}

\subsubsection{Group $\mathcal{D}_2$}
Only $Q^2$ and $M_1^2$ are nonvanishing. Thus, the quasi-entropy is 
\begin{equation}
  q_2(Q^2,M_1^2)=-\ln\det (\frac{I_3}{3}+Q^2)-\ln\det(\frac{I_3}{3}-\frac{Q^2}{2}+M_1^2)-\ln\det(\frac{I_3}{3}-\frac{Q^2}{2}-M_1^2). \label{D2}
\end{equation}

\subsection{Quasi-entropy based on $q^{(0)}_4$}
Next, we discuss the quasi-entropy for some groups allowing three- or four-fold rotations. 
Since they involve tensors of third and fourth order, we need to consider $q_4^{(0)}$ from $\bm{w}_2$. 
Choose an orthogonal basis of second order symmetric traceless tensors, 
\begin{align*}
  &W^2_1(\mathfrak{p})=\bm{m}_1^2-\mathfrak{i}/3,\
  W^2_2(\mathfrak{p})=\frac{1}{2}(\bm{m}_2^2-\bm{m}_3^2),\\
  &W^2_3(\mathfrak{p})=\bm{m}_1\bm{m}_2,\
  W^2_4(\mathfrak{p})=\bm{m}_1\bm{m}_3,\
  W^2_5(\mathfrak{p})=\bm{m}_2\bm{m}_3.
\end{align*}
The vector $\bm{w}_2$ has $34$ components, consisting of $\bm{e}_i\cdot\bm{m}_j$ and $W^2_i\cdot W^2_j(\mathfrak{p})$. 

The calculation of covariance matrix is cumbersome, so we here write down the results only and present the derivations in Appendix. 
Let us introduce a notation of a $3\times 5$ matrix from a third order tensor, 
\begin{align*}
  \Psi_3(V)=\left(
  \begin{array}{ccccc}
    V_{111} & \frac{1}{2}(V_{122}-V_{133}) & V_{112} & V_{113} & V_{123}\\
    V_{211} & \frac{1}{2}(V_{222}-V_{233}) & V_{212} & V_{213} & V_{223}\\
    V_{311} & \frac{1}{2}(V_{322}-V_{333}) & V_{312} & V_{313} & V_{323}
  \end{array}
  \right), 
\end{align*}
where $V_{ijk}$ is the $(i,j,k)$-component of the tensor $V$. 
For a fourth order tensor, we define a $5\times 5$ matrix as 
\begin{align*}
  \Psi_4(V)=\left(
  \begin{array}{lllll}
    V_{1111} & \frac{1}{2}(V_{1122}-V_{1133}) & V_{1112} & V_{1113} & V_{1123}\vspace{6pt}\\
    \frac{1}{2}(V_{2211} & \frac{1}{4}(V_{2222}-V_{2233} & \frac{1}{2}(V_{2212} & \frac{1}{2}(V_{2213} & \frac{1}{2}(V_{2223}\\
    \quad-V_{3311}) & \quad-V_{3322}+V_{3333}) & \quad-V_{3312}) & \quad-V_{3313}) & \quad-V_{3323})\vspace{6pt}\\
    V_{1211} & \frac{1}{2}(V_{1222}-V_{1233}) & V_{1212} & V_{1213} & V_{1223}\\
    V_{1311} & \frac{1}{2}(V_{1322}-V_{1333}) & V_{1312} & V_{1313} & V_{1323}\\
    V_{2311} & \frac{1}{2}(V_{2322}-V_{2333}) & V_{2312} & V_{2313} & V_{2323}
  \end{array}
  \right). 
\end{align*}

\subsubsection{Octahedral symmetry, group $\mathcal{O}$}
Up to fourth order, the nonvanishing tensor is 
\begin{align}
  O=&\langle\bm{m}_1^2\bm{m}_2^2+\bm{m}_2^2\bm{m}_3^2+\bm{m}_1^2\bm{m}_3^2-\frac{1}{5}\mathfrak{i}^2\rangle. \label{Oten}
\end{align}
Denote
\begin{align}
  D=\mathrm{diag}(\frac{1}{15},\frac{1}{20},\frac{1}{20},\frac{1}{20},\frac{1}{20}). \label{Dmat}
\end{align}
The quasi-entropy is given by 
\begin{align}
  q_4(O)=-2\ln\det\big(D-\frac{1}{2}\Psi_4(O)\big)-3\ln\det\big(D+\frac{1}{3}\Psi_4(O)\big). \label{qoct}
\end{align}

\subsubsection{Tetrahedral symmetry, group $\mathcal{T}$}
The nonvanishing tensors include the fourth order tensor $O$, and a third order tensor 
\begin{align}
  T=&\langle\bm{m}_1\bm{m}_2\bm{m}_3\rangle. \label{Tten}
\end{align}
Define 
\begin{align}
  (A_1)_{ijkl}=\sum_{s=1}^3\epsilon_{jks}T_{ils}+\epsilon_{ils}T_{jks}. 
\end{align}
The quasi-entropy is given by
\begin{align}
  q_4(T,O)=&-\ln\det\left(
  \begin{array}{cc}
    \frac{4}{3}D-\frac{2}{3}\Psi_4(O) & \frac{1}{2}\Psi_4(A_1)\\
    \frac{1}{2}\Psi_4(A_1)^t & D-\frac{1}{2}\Psi_4(O)
  \end{array}
  \right)\nonumber\\
  &-3\ln\det\left(
  \begin{array}{cc}
    \frac{1}{3}I_3 & \Psi_3(T)\\
    \Psi_3(T)^t & D+\frac{1}{3}\Psi_4(O)
  \end{array}
  \right). \label{qtet}
\end{align}

\subsubsection{Four-fold dihedral group $\mathcal{D}_4$}
The nonvanishing tensors include $Q^2$, and
\begin{align}
  &Q^4=\langle\bm{m}_1^4-\frac{6}{7}\bm{m}_1^2\mathfrak{i}+\frac{3}{35}\mathfrak{i}^2\rangle, \\
  &M^4_1=\langle 8\bm{m}_2^4-8(\mathfrak{i}-\bm{m}_1^2)\bm{m}_2^2+(\mathfrak{i}-\bm{m}_1^2)^2\rangle.
\end{align}
Define
\begin{align}
  &(A_2)_{ijkl}=\delta_{kl}Q^2_{ij}+\delta_{ij}Q^2_{kl}-\frac{3}{4}(\delta_{ik}Q^2_{jl}+\delta_{jl}Q^2_{ik}+\delta_{il}Q^2_{jk}+\delta_{jk}Q^2_{il}),\\
  &(B_1)_{ijk}=\sum_{s=1}^3\epsilon_{ijs}Q^2_{ks}+\epsilon_{iks}Q^2_{js}. 
\end{align}
The quasi-entropy is given by 
\begin{align}
  q_4(Q^2,Q^4,M_1^4)=&-\ln\det(Q^2+\frac{1}{3}I_3)
  -\ln\det\big(\Psi_4(Q^4-\frac{4}{21}A_2-Q^2\otimes Q^2)+\frac{4}{3}D\big)\nonumber\\
  &-\ln\det\big(\Psi_4(\frac{1}{8}Q^4+\frac{1}{8}M_1^4+\frac{1}{7}A_2)+D\big)\nonumber\\
  &-\ln\det\big(\Psi_4(\frac{1}{8}Q^4-\frac{1}{8}M_1^4+\frac{1}{7}A_2)+D\big)\nonumber\\
  &-\ln\det\left(
  \begin{array}{cc}
    -\frac{1}{2}Q^2+\frac{1}{3}I_3 & \frac{1}{4}\Psi_3(B_1)\\
    \frac{1}{4}\Psi_3(B_1)^t & \Psi_4(-\frac{1}{2}Q^4-\frac{1}{14}A_2)+D
  \end{array}
  \right)\nonumber\\
  &-\ln\det\left(
  \begin{array}{cc}
    -\frac{1}{2}Q^2+\frac{1}{3}I_3 & -\frac{1}{4}\Psi_3(B_1)\\
    -\frac{1}{4}\Psi_3(B_1)^t & \Psi_4(-\frac{1}{2}Q^4-\frac{1}{14}A_2)+D
  \end{array}
  \right). 
\end{align}

\subsubsection{Three-fold dihedral group $\mathcal{D}_3$}
The nonvanishing tensors include $Q^2$, $Q^4$, and 
\begin{align}
  &M_1^3=\langle 4\bm{m}_2^3-3(\mathfrak{i}-\bm{m}_1^2)\bm{m}_2\rangle,\\
  &N^4=\langle 4\bm{m}_1\bm{m}_2^2\bm{m}_3-\bm{m}_1(\mathfrak{i}-\bm{m}_1^2)\bm{m}_3\rangle.
\end{align}
Define
\begin{align}
  (B_2)_{ijkl}=\sum_{s=1}^3\epsilon_{iks}(M_1^3)_{jls}+\epsilon_{jls}(M_1^3)_{iks}. 
\end{align}
The quasi-entropy is given by 
\begin{align}
  q_4(Q^2,&M_1^3,Q^4,M_2^4)\nonumber\\
  =&-\ln\det(Q^2+\frac{1}{3}I_3)
  -\ln\det\big(\Psi_4(Q^4-\frac{4}{21}A_2-Q^2\otimes Q^2)+\frac{4}{3}D\big)\nonumber\vspace{6pt}\\
  &-\ln\det\left(
  \begin{array}{ccc}
    -\frac{1}{2}Q^2+\frac{1}{3}I_3 & \frac{1}{4}\Psi_3(M_1^3) & \frac{1}{4}\Psi_3(B_1)\vspace{6pt}\\
    \frac{1}{4}\Psi_3(M_1^3)^t & \Psi_4(\frac{1}{8}Q^4+\frac{1}{7}A_2)+D & \frac{1}{4}\Psi_4(N^4+B_2)\vspace{6pt}\\
    \frac{1}{4}\Psi_3(B_1)^t & \frac{1}{4}\Psi_4(N^4+B_2)^t & \Psi_4(-\frac{1}{2}Q^4-\frac{1}{14}A_2)+D
  \end{array}
  \right)\nonumber\\
  &-\ln\det\left(
  \begin{array}{ccc}
    -\frac{1}{2}Q^2+\frac{1}{3}I_3 & -\frac{1}{4}\Psi_3(B_1) & -\Psi_3(M_1^3)\vspace{6pt}\\
    -\frac{1}{4}\Psi_3(B_1)^t & \Psi_4(-\frac{1}{2}Q^4-\frac{1}{14}A_2)+D & \frac{1}{4}\Psi_4(N^4-B_2)\vspace{6pt}\\
    -\frac{1}{4}\Psi_3(M_1^3)^t & \frac{1}{4}\Psi_4(N^4-B_2)^t & \Psi_4(\frac{1}{8}Q^4+\frac{1}{7}A_2)+D
  \end{array}
  \right).
\end{align}

\section{Applications\label{appl}}
In this section, we discuss some applications of the quasi-entropy to liquid crystals, for which we restrain our attention to spatially homogeneous systems. 
As we have described at the beginning, we consider the free energy consisting of the quasi-entropy and a polynomial. 
Generally, if the free energy is a function of the tensors $\langle U_j(\mathfrak{p})\rangle=\avrg{U}_j$, it is written in the following form, 
\begin{align}
  f(\avrg{U}_j)=\nu q_{2n}(\avrg{U}_j)+f_1(\avrg{U}_j). \label{genfe}
\end{align} 
A coefficient $\nu>0$ is introduced before the quasi-entropy, because for any $\nu$ the function satisfies the properties we discussed.
The value of $\nu$ can be determined under certain principle, as we will discuss below. 
In the second term, we only include the pairwise interaction term given by a quadratic polynomial of $\avrg{U}_j$.

We are interested in the stationary points of the free energy.
The local minimizers represent different phases, and the saddle points might stand for the transition states between the phases.

Several examples will be presented to illustrate the usefulness of quasi-entropy. 
On one hand, we examine some significant cases that have been discussed using the original entropy.
We will see that using the quasi-entropy leads to results consistent with using the original entropy. 
On the other hand, we present some results that have not been reported before. 
These examples give strong evidence on the capability of quasi-entropy to capture the essence of the orginal entropy, as well as its simplicity and versatility. 

\subsection{Rod-like molecules, group $\mathcal{D}_{\infty}$}
Only the tensor $Q=\langle\bm{m}_1^2-\mathfrak{i}/3\rangle$ is involved. Consider the free energy 
\begin{align}
  f(Q)=\nu(-\ln\det Q-2\ln\det\frac{I_3-Q}{2})-\frac{\eta}{2}|Q|^2. \label{oneQ}
\end{align}
For the convenience of analysis, we introduce $R_1=Q+I_3/3=\langle\bm{m}_1^2\rangle$. 

\subsubsection{Axisymmetry of the stationary points}
A well-known mathematical result when using the original entropy is that the stationary points must have two equal eigenvalues \cite{liu2005axial,fatkullin2005critical}. 
We show it when using the quasi-entropy. 
\begin{theorem}
The stationary points of \eqref{oneQ} must have at least two equal eigenvalues. 
\end{theorem}
\begin{proof}
  By rotational invariance (Theorem \ref{qprop}), we may assume that $R_1=Q+I_3/3=\mbox{diag}(x_1,x_2,x_3)$ is diagonal, where $0<x_i<1$ and $x_1+x_2+x_3=1$. Denote $\chi=\eta/\nu$. Now the energy is given by 
  \begin{equation}
    \frac{f(Q)}{\nu}=\frac{1}{6}\chi+\sum_{i=1}^3 -\ln x_i-2\ln\frac{1-x_i}{2}-\frac{\chi}{2}x_i^2. 
  \end{equation}
  The Euler-Lagrange equation is written as 
  %
  \begin{equation}
    \frac{\partial F}{\partial x_i}=-\frac{1}{x_i}+\frac{2}{1-x_i}-\chi x_i=\mu, 
  \end{equation}
  where $\mu$ is a Lagrange multiplier that makes $\sum x_i=1$.
  The above equation implies that $x_i$ are the roots of the polynomial 
  \begin{equation}
    \chi x^3-(\chi-\mu)x^2+(3-\mu)x-1=0. 
  \end{equation}
  If $x_i$ are mutually unequal, they must be the three distinct roots of the above cubic polynomial, and shall lie within $(0,1)$. Thus, we must have $1-\mu/\chi=1$, indicating that $\mu=0$. In this case, we derive that 
  \begin{equation}
    \chi=\frac{3x_i-1}{x_i{}^2(1-x_i)}. 
  \end{equation}
  However, there must be one $x_i>1/3$ that makes $\chi>0$, and another $x_i<1/3$ that makes $\chi<0$, which is impossible. 
\end{proof}

\subsubsection{The coefficient $\nu$}
By the above theorem, we may let
\begin{equation}
  Q+\frac{1}{3}I_3=R_1=\mbox{diag}(x,\frac{1-x}{2},\frac{1-x}{2}). \label{uniaxl}
\end{equation}
Under this assumption, the quasi-entropy is given by 
\begin{align*}
  q_2(x)=\nu\left(-\ln x-4\ln\frac{1-x}{2}-4\ln\frac{1+x}{4}\right). 
\end{align*}
We could compare this energy with the one including the original entropy $f_{\mathrm{ent}}(Q)$, to propose a reasonable value of $\nu$. 
When $Q+I_3/3$ is given by \eqref{uniaxl}, the original entropy is given by 
\begin{align*}
  f_{\mathrm{ent}}(x)=bx-\ln Z=bx-\ln\frac{1}{2}\int_{-1}^1\md z\,\exp(bz^2),\quad x=\frac{1}{Z}\cdot\frac{\partial Z}{\partial b}. 
\end{align*}
For both of the original entropy $f_{\mathrm{ent}}$ and the quasi-entropy $q_2$, the minimizer is $x=1/3$. 
We choose $\nu$ such that the second derivative about $x$ of two functions are equal at $x=1/3$. 
Let us calculate 
\begin{align*}
  \frac{\partial f_{\mathrm{ent}}}{\partial x}=&b+x\frac{\partial b}{\partial x}-\frac{1}{Z}\frac{\partial Z}{\partial x}
  =b+x\frac{\partial b}{\partial x}-\frac{1}{Z}\frac{\partial Z}{\partial b}\frac{\partial b}{\partial x}\\
  =&b+x\frac{\partial b}{\partial x}-x\frac{\partial b}{\partial x}\\
  =&b. 
\end{align*}
Thus, when $x=1/3$, we have $b=0$. The second derivative of the original entropy is calculated as 
\begin{align*}
  \frac{\partial^2f_{\mathrm{ent}}}{\partial x^2}\bigg|_{b=0}=&\left(\frac{\partial b}{\partial x}\right)\bigg|_{b=0}
  =\left(\frac{\partial x}{\partial b}\right)^{-1}\bigg|_{b=0}\\
  =&\left(\frac{1}{Z}\cdot\frac{\partial^2 Z}{\partial b^2}-\left(\frac{1}{Z}\cdot\frac{\partial Z}{\partial b}\right)^2\right)^{-1}\bigg|_{b=0}\\
  =&\left(\int_{-1}^1\md z\,z^4\frac{\exp(bz^2)}{Z}-\left(\int_{-1}^1\md z\,z^2\frac{\exp(bz^2)}{Z}\right)^2\right)^{-1}\bigg|_{b=0}\\
  =&\frac{45}{4}. 
\end{align*}
For $q_2(x)$, the second derivative at $x=1/3$ is $\partial^2 q_2/\partial x^2\big|_{x=1/3}=81\nu/4$. 
Hence, $\nu=5/9$ would be a reasonable choice. 

\subsubsection{Stationary points}
With \eqref{uniaxl}, the free energy becomes 
\begin{equation}
  f(x)=\nu(-\ln x-4\ln\frac{1-x}{2}-4\ln\frac{1+x}{4})-\frac{\eta}{2}\Big(x^2+\frac{(1-x)^2}{2}-\frac{1}{3}\Big). \label{oneQ2}
\end{equation}

\begin{theorem}\label{rodsols}
The stationary points of \eqref{oneQ2} are characterized by $\chi=\eta/\nu$. 
Let $\chi_1/2=\frac{3x+1}{x(1-x^2)}\Big|_{x=x^*}\approx 6.532952$, where $x^*$ is the real root of $6x^3+3x^2=1$, and $\chi_2/2=6.75$. 
\begin{enumerate}
\item When $\chi<\chi_1$, there is one solution $1/3$, which is the minimizer. 
\item When $\chi=\chi_1$, there are two solutions: $1/3$ is the minimizer, and $x^*$ is a saddle point. 
\item When $\chi_1<\chi<\chi_2$, there are three solutions $1/3<x_2<x_3$, where $1/3$ and $x_3$ are two local minimizers, and $x_2$ is a local maximizer. 
\item When $\chi=\chi_2$, there are two solutions: $1/3$ is a saddle point, and a minimizer $x_3>1/3$. 
\item When $\chi>\chi_2$, there are three solutions $x_1<1/3<x_3$, where $1/3$ is a local maximizer, and the other two are local minimizers. 
\end{enumerate}
\end{theorem}
\begin{proof}
  Notice that $0<x<1$. 
  We compute the first, second and fourth derivatives of $f$ about $x$:
  \begin{align*}
    \frac{1}{\nu}f'=&(3x-1)\Big(\frac{3x+1}{x(1-x^2)}-\frac{\chi}{2}\Big),\\
    \frac{1}{\nu}f''=&\frac{1}{x^2}+\frac{4}{(1-x)^2}+\frac{4}{(1+x)^2}-\frac{3\chi}{2},\\
    \frac{1}{\nu}f''''=&\frac{6}{x^4}+\frac{12}{(1-x)^4}+\frac{12}{(1+x)^4}>0. 
  \end{align*}
  As a result of $f''''>0$, the second derivative $f''$ can have at most two zeros in the interval $(0,1)$. It further indicates that $f'$ can have at most three zeros. 
  
  The derivative of the second parentheses in $f'$ is calculated as 
  \begin{align*}
    \Big(\frac{3x+1}{x(1-x^2)}\Big)'=\frac{6x^3+3x^2-1}{(x-x^3)^2}. 
  \end{align*}
  When $0<x<1$, its unique zero point is $x^*$ that is greater than $1/3$. The derivative is positive when $x>x^*$, and negative when $x<x^*$. 
  Hence, the function $\frac{3x+1}{x(1-x^2)}$ reaches the minimum at $x=x^*$. 
  
  Therefore, when $\chi<\chi_1$, the function $\frac{3x+1}{x(1-x^2)}-\chi/2$ is positive, which implies that $f'$ has only one zero point $x=1/3$. 
  When $\chi=\chi_1$, the function $\frac{3x+1}{x(1-x^2)}-\chi/2$ has exactly one zero point $x^*$, indicating that $f'$ has exactly two zero points $1/3$ and $x^*$. 
  When $\chi>\chi_1$, the function $\frac{3x+1}{x(1-x^2)}-\chi/2$ has exactly two zero points.
  It is possible that $1/3$ is a zero point of $\frac{3x+1}{x(1-x^2)}-\chi/2$. This case occurs if and only if $\chi=\chi_2$.
  Thus, if $\chi_1<\chi<\chi_2$ or $\chi>\chi_2$, the first derivative $f'$ has three zero points. 
  When $\chi=\chi_2$, there are only two zero points. 
  
  Note that 
  \begin{align*}
    &\lim_{x\to 0^+}f'=-\infty,\ \lim_{x\to 1^-}f'=+\infty,\\
    &\lim_{x\to 0^+}f''=\lim_{x\to 1^-}f''=+\infty. 
  \end{align*}
  In the cases where $f'$ has three zero points $x_1<x_2<x_3$, one zero point of $f''$ must lie between $x_1$ and $x_2$, and the other between $x_2$ and $x_3$.
  Therefore, $x_1$ and $x_3$ are local minimizers of $f$, while $x_2$ is a local maximizer. 
  When $\chi_1<\chi<\chi_2$, we have $f''(1/3)>0$ and $f'(x^*)<0$, so that $x_1=1/3<x_2<x_3$. 
  When $\chi>\chi_2$, we have $f''(1/3)<0$. It implies that there exists $y_1<1/3,\,y_2>1/3$ such that $f'(y_1)>0,\,f'(y_2)<0$, so that $x_1<x_2=1/3<x_3$.
  It remains to examine the cases $\chi=\chi_1$ and $\chi=\chi_2$.
  
  When $\chi=\chi_1$, we have $f''(1/3)>0$. Meanwhile, there is only one zero point of $f'(x^*)=0$ other than $1/3$. Together with $f'\to+\infty$ when $x\to 1^-$, we deduce that $f'>0$ for $x>1/3,\, x\ne x^*$. 
  Therefore, $1/3$ is a local minimizer of $f$, and $x^*$ is a saddle point. 
  
  When $\chi=\chi_2$, we have $f'(1/3)=f''(1/3)=0$ and $f'(x^*)<0$. Thus, we can find some $y>1/3$ such that $f''(y)<0$. Let $y_1>1/3$ such that $f''(y_1)=0$. Since $f''$ has exactly two zero points, we must have $f''(y)<0$ for $1/3<y<y_1$ and $f''(y)>0$ for $y<1/3$ and $y>y_1$. Thus, $1/3$ is a saddle point. Then, from $f'(y_1)=\int_{1/3}^{y_1}f''(y)\md y<0$, we deduce that the zero point of $f'$ other than $1/3$ is greater than $y_1$. Hence, at this zero point of $f'$, the second derivative is positive, indicating that it is a local minimizer. 
  
\end{proof}

We compare the result with that with the original entropy \cite{liu2005axial}. 
When $\chi$ is increasing, the changes in the type of stationary points are identical. 
No matter we use the original entropy or the quasi-entropy, we could see that as $\eta$ is increasing, a minimizer other than $x=1/3$ emerges, followed by $x=1/3$ losing its stability.
This exactly describes the first order phase transition from the isotropic phase to the uniaxial nematic phase. 

Theorem \ref{rodsols} gives two critical values of $\chi=\eta/\nu$. 
If we adopt the value $\nu=5/9$, the two critical values are $\eta=7.258835$ and $\eta=7.5$. 
When using the original entropy, the two critical values, given in terms of $\eta$, are $\eta=6.731393$ and $\eta=7.5$.
The value $\eta=7.5$ are identical because we choose $\nu$ according to the second derivative at $x=1/3$. 
Although the specific value of another critical value is not identical, the structure of the solutions is the same, which is sufficient to identify the underlying physics in the free energy. 

\subsection{Two-fold symmetry, group $\mathcal{D}_2$}
Again, we introduce some alternative notations of tensors for our analysis. 
Define $R_i=\langle\bm{m}_i^2\rangle$. 
Recall that the quasi-entropy \eqref{D2} is expressed by $Q^2=\langle\bm{m}_1^2-\mathfrak{i}/3\rangle$ and $M_1^2=\frac{1}{2}\langle\bm{m}_2^2-\bm{m}_3^2\rangle$. 
It can be recognized that 
\begin{align}
  R_1=Q^2+\frac{1}{3}I_3,\quad R_2=\frac{1}{3}I_3-\frac{1}{2}Q^2+M_1^2,\quad R_3=\frac{1}{3}I_3-\frac{1}{2}Q^2-M_1^2=I_3-R_1-R_2. \label{Qs}
\end{align}
Each of the three tensors $R_i$ has positive eigenvalues and the trace one. 
It is equivalent to use any two of $R_i$ instead of $Q^2$ and $M_1^2$.

To simplify our derivation below, we express the free energy in terms of $R_1$ and $R_2$, 
\begin{align*}
  f(R_1,R_2)
  =&\nu(-\ln\det R_1-\ln\det R_2-\ln\det R_3)\nonumber\\
  &+\frac{\eta}{2}(c_{02}|R_1|^2+c_{03}|R_2|^2+2c_{04}R_1\cdot R_2), 
\end{align*}
where the coefficients $c_{02},\,c_{03},\,c_{04}$ in the pairwise interaction term can be derived from molecular parameters \cite{Bi1,rosso2006quadrupolar,SymmO,BentModel}. 
Moreover, we rewrite the coupling term $R_1\cdot R_2$ using $R_3$, yielding 
\begin{align}
  \frac{f(R_1,R_2)}{\nu}=&-\ln\det R_1-\ln\det R_2-\ln\det R_3\nonumber\\
  &+\frac{\eta}{2\nu}\Big((c_{02}-c_{04})|R_1|^2+(c_{03}-c_{04})|R_2|^2+c_{04}|R_3|^2\Big)+\text{Constant}. \label{D2eng}
\end{align}
Let us denote $c_1=-\eta(c_{02}-c_{04})/\nu$, $c_2=-\eta(c_{03}-c_{04})/\nu$, $c_3=-\eta c_{04}/\nu$. 

\subsubsection{Shared eigenframe of $R_i$}
When using the $f_{\mathrm{ent}}$, under some conditions on $c_j$, at stationary points the eigenframes of $R_i$ are identical \cite{xu2017transmission}.
We show an analogous result with quasi-entropy. 
\begin{theorem}\label{D2diag}
For the stationary points of \eqref{D2eng}, if either of $c_1$, $c_2$, $c_3$ is no greater than $4$, then $R_1$, $R_2$ and $R_3$ have the same eigenframe. 
\end{theorem}
\begin{proof}
  %
  The Euler-Lagrange equation is given by 
  \begin{equation}
    R_1^{-1}+c_1R_1+\lambda_1I=R_2^{-1}+c_2R_2+\lambda_2I=R_3^{-1}+c_3R_3+\lambda_3I, 
  \end{equation}
  where $\lambda_i$ are Lagrange multipliers making $\mathrm{tr}R_i=1$. 
  Without losing generality, we assume $c_2\le 4$. 
  We can deduce from the above that 
  \begin{align*}
    R_1(R_2^{-1}+c_2R_2+\lambda_2I)=(R_2^{-1}+c_2R_2+\lambda_2I)R_1, 
  \end{align*}
  yielding
  \begin{align*}
    R_2R_1-R_1R_2=c_2R_2(R_2R_1-R_1R_2)R_2. 
  \end{align*}
  By rotational invariance, we can assume that $R_2=\mbox{diag}(z_1,z_2,z_3)$. 
  Denote $R=R_2R_1-R_1R_2$. Computing the off-diagonal entries in the above equation, we obtain
  \begin{align*}
    R_{ij}(z_i-z_j)(1-c_2z_iz_j)=0. 
  \end{align*}
  Because $R_2$ is positive definite and trace-one, we have $z_iz_j<1/4$. Along with $c_2\le 4$, we have $1-c_2z_iz_j>0$, leading to 
  \begin{align*}
    R_{ij}(z_i-z_j)=0. 
  \end{align*}
  If $z_i\ne z_j$, we have $R_{ij}=0$.
  If $z_i=z_j$, we calculate that
  $$
  (R_2R_1)_{ij}=z_i(R_1)_{ij},\quad (R_1R_2)_{ij}=z_j(R_1)_{ij}, 
  $$
  which also implies that $R_{ij}=0$.
  Therefore, we always have $R_1R_2=R_2R_1$, so that $R_1$ and $R_2$ share an eigenframe. 
\end{proof}
Let us compare the condition in Theorem \ref{D2diag} with that in Theorem 1.1 in \cite{xu2017transmission}. 
If expressed by $c_1,c_2,c_3$ in the current work, the condition in \cite{xu2017transmission} can be rewritten as: 
\begin{enumerate}[(a)]
\item The matrix $\displaystyle\left(\begin{array}{cc}c_1+c_3 & c_3\\ c_3 & c_2+c_3\end{array}\right)$ is not positive definite; OR
\item This matrix is positive definite, but $-\frac{c_3^2}{c_1+c_3}+c_2+c_3\le \frac{2}{\nu}$.
\end{enumerate}
When taking $\nu=5/9$ as we suggest, it is easy to verify that when $c_1,c_2,c_3>4$, the above condition cannot hold. Thus, the condition in Theorem \ref{D2diag} is necessary, but not sufficient, for the condition above.
On the other hand, some coefficients derived from molecular parameters are checked in \cite{xu2017transmission}.
It turns out that all of these coefficients satisfy the condition in Theorem \ref{D2diag}. 

\subsubsection{Stationary points without shared eigenframe}
It has not been answered in \cite{xu2017transmission} whether there exists stationary points with $R_i$ having different eigenframes if the condition on $c_1,c_2,c_3$ does not hold. 
Here, with the quasi-entropy, the answer is yes. 
In the following, we construct a class of examples. 
Let $0<a<1/2$, $c,r>0$, and 
$$
c_1=\frac{1}{a^2-c^2}>0,\quad c_2=c_3=\frac{1}{\frac{1}{4}(1-a)^2-\frac{1}{4}c^2-r^2}>0. 
$$
We consider a stationary point of the following form, 
\begin{align*}
R_1=&\left(
\begin{array}{ccc}
  a+c & 0 & 0 \\
  0 & a-c & 0 \\
  0 & 0 & 1-2a
\end{array}
\right),\\ 
R_2=&\left(
\begin{array}{ccc}
  \frac{1}{2}(1-a-c) & r & 0 \\
  r & \frac{1}{2}(1-a+c) & 0 \\
  0 & 0 & a
\end{array}
\right),\\ 
R_3=&\left(
\begin{array}{ccc}
  \frac{1}{2}(1-a-c) & -r & 0 \\
  -r & \frac{1}{2}(1-a+c) & 0 \\
  0 & 0 & a
\end{array}
\right). 
\end{align*}
Then we have 
\begin{align*}
&R_1^{-1}+c_1R_1=\mbox{diag}\Big(2ac_1,2ac_1,\frac{1}{1-2a}+c_1(1-2a)\Big), \\ 
&R_2^{-1}+c_2R_2=R_3^{-1}+c_3R_3=\mbox{diag}\Big((1-a)c_2,(1-a)c_2,\frac{1}{a}+c_2a\Big). 
\end{align*}
By the Euler-Lagrange equation, we have 
$$
c_2(1-a)-2c_1a=\frac{1}{a}+c_2a-\frac{1}{1-2a}-c_1(1-2a). 
$$
Let $\lambda=a^2-c^2\in (0, a^2]$. Then we have 
$$
\frac{a(1-2a)^2\lambda}{(1-3a)\lambda+a(1-2a)(4a-1)}=\frac{1-2a+\lambda}{4}-r^2\le \frac{1-2a+\lambda}{4}. 
$$
It gives an inequality about $a$ and $\lambda=a^2-c^2$.
If the inequality holds, we can always solve $r$. 
Thus, let us focus on the inequality. 
\begin{enumerate}
\item $a=1/3$. The above inequality gives $c^2\ge 0$. So we have $0\le c<a$. 
\item $1/3<a<1/2$. When $a\ne 1/3$, the inequality can be converted into 
$$
(\lambda-\lambda_0)(\lambda-\lambda_1)(\lambda-\lambda_2)\ge 0,\quad \lambda\ne\lambda_0, 
$$
with 
$$
\lambda_0=\frac{a(1-2a)(4a-1)}{3a-1},\ \lambda_1=\frac{a(1-2a)}{1-3a},\ \lambda_2=(1-2a)(4a-1). 
$$
When $1/3<a<1/2$, we have $\lambda_1<0<\lambda_2<\lambda_0$. 
So the inequality gives $0<\lambda\le (1-2a)(4a-1)$. 
From this, we solve that $c\ge 3a-1$. Thus, we have $3a-1\le c<a$. 
\item $1/4<a<1/3$. In this case, we can verify that $\lambda_1>\lambda_2>0>\lambda_0$. 
So we have $a^2\ge\lambda\ge \lambda_1$ or $0<\lambda\le \lambda_2$. 
We can check that $a^2<\lambda_1$ for any $0<a<1/3$, so we can only have $0<\lambda\le \lambda_2$, which gives $1-3a\le c<a$. 
\item $0<a\le 1/4$. Now we have $\lambda_1>0\ge \lambda_2,\lambda_0$. However, we still have $a^2<\lambda_1$ because it holds for any $0<a<1/3$. 
Therefore, when $a\le 1/4$ such a stationary point does not exists. 
\end{enumerate}
To summarize, the stationary point given in this paragraph exists for $1/4<a<1/2$ and $|1-3a|\le c<a$. 

\subsection{Bent-core molecules, group $\mathcal{C}_2$}
A free energy involving three tensors is proposed in \cite{SymmO,BentModel}, which includes the original entropy $f_{\mathrm{ent}}$. 
The three tensors are equivalent to $Q^1,Q^2,M_1^2$, where we recall $Q^1=\langle\bm{m}_1\rangle$. We still use the notations $R_i$ in \eqref{Qs}. 
When the original entropy is substituted by the quasi-entropy, the free energy is written as 
\begin{align*}
  f(Q^1,Q^2,M_1^2)
  =&\nu q_2(Q^1,Q^2,M_1^2)\nonumber\\
  &+\frac{\eta}{2}(c_{01}|Q^1|^2+c_{02}|R_1|^2+c_{03}|R_2|^2+2c_{04}R_1\cdot R_2). 
\end{align*}
The quasi-entropy about three tensors is given in \eqref{C2simp}. 
The coefficients can be calculated as functions of molecular shape parameters \cite{SymmO,BentModel,xu2018onsager-theory-based}.  
We will focus on the effect of varying bending angle $\theta$ (see Fig. \ref{qhom}). 

With the coefficients derived from molecular parameters, the coefficient $c_{01}$ turns out to be positive. 
Using the third property in Theorem \ref{qprop}, when $Q^2$ and $M_1^2$ are fixed, $q_2(Q^1,Q^2,M_1^2)$ has the unique minimizer $Q^1=0$. 
As a result, the free energy is reduced to \eqref{D2eng}. 
Furthermore, we could verify that the conditions in Theorem \ref{D2diag} indeed hold, so we can assume that $R_i$ are all diagonal. 

We would like to examine the phase diagram of bent-core molecules, about $\eta$ and $\theta$. 
The coefficient before the quasi-entropy is chosen as $\nu=5/9$ from the discussion of rod-like molecules. 
We solve the local minimizers numerically by steepest descend method, using multiple initial values in order not to miss any. 
Then, we pick the one that has the lowest energy and label it as the one to appear in the phase diagram. 
The phases are characterized by the eigenvalues of $R_i$ (see \cite{BentModel}). 
\begin{itemize}
\item Isotropic phase: $R_1=R_2=R_3=I_3/3$. 
\item Uniaxial nematic phase $N_i$ ($i=1,2,3$): The tensors $R_j=(a_j,(1-a_j)/2,(1-a_j)/2)$ are uniaxial for $j=1,2,3$. In the phase $N_i$, we have $a_i>1/3$ but the other two $a_j<1/3$. 
\item Biaxial phase $B$: each $R_i$ has three distinct eigenvalues. 
\end{itemize}
We draw the phase diagram in Fig. \ref{qhom}, comparing the ones using the original entropy and the quasi-entropy. 
It turns out that they are very close.
\begin{figure}
\centering
\includegraphics[width=.3\textwidth,keepaspectratio]{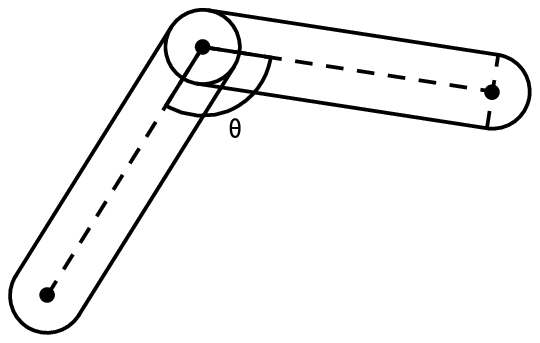}
\includegraphics[width=.5\textwidth,keepaspectratio]{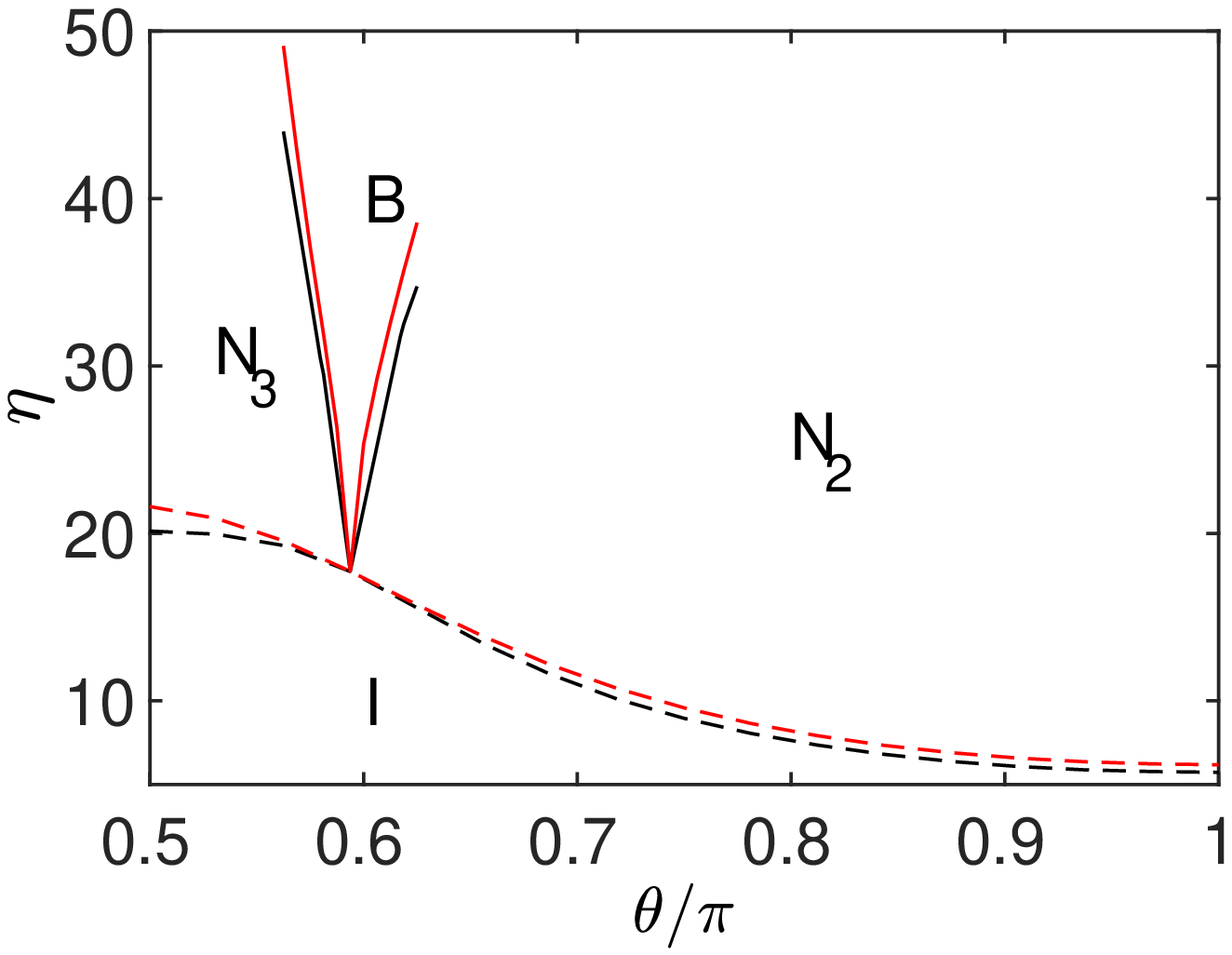}
\caption{Comparison of phase diagram of bent-core molecules. The black lines represent the one with the original entropy, the red lines represent the one with quasi-entropy.}\label{qhom}
\end{figure}

\subsection{Tetrahedral and octahedral symmetries}
For the group $\mathcal{T}$, two tensors $T$ and $O$ are involved, which are defined in \eqref{Tten} and \eqref{Oten}. 
If the interaction term is quadratic, we could write down the free energy using the quasi-entropy as follows (cf. \cite{xu_kernel} for interaction terms) 
\begin{align}
  f(T,O)=q_4(T,O)-\frac{1}{2}\mu_1|T|^2-\frac{1}{2}\mu_2|O|^2, \label{Teng}
\end{align}
where $q_4(T,O)$ is given in \eqref{qtet}.
If $\mu_1=0$, by $\min_Tq_4(T,O)=q_4(O)$, \eqref{Teng} is reduced to a free energy for the group $\mathcal{O}$, 
\begin{align}
  f(O)=q_4(O)-\frac{1}{2}\mu_2|O|^2, \label{Oeng}
\end{align}
with $q_4(O)$ given in \eqref{qoct}.


As a simplification, we assume that the two tensors can be written as 
\begin{align*}
  T=s\bm{e}_1\bm{e}_2\bm{e}_3,\quad O=t(\bm{e}_1^2\bm{e}_2^2+\bm{e}_3^2\bm{e}_3^2+\bm{e}_3^2\bm{e}_1^2-\frac{1}{5}\mathfrak{i}^2), 
\end{align*}
where a polynomial of $\bm{e}_i$ represents a symmetric tensor (see \eqref{tensor_monomial}).
Such an assumption restrains our attention to the phases having tetrahedral symmetry (see \cite{xu_tensors}). 
The nonzero components in the two tensors are given by 
\begin{align*}
  &T_{123}=\frac{1}{6}s,\\
  &O_{1122}=O_{1133}=O_{2233}=\frac{1}{10}t,\\
  &O_{1111}=O_{2222}=O_{3333}=-\frac{1}{5}t. 
\end{align*}
Taking them into the quasi-entropy, we arrive at
\begin{align}
  q_4(s,t)=&-2\ln\Big((\frac{1}{15}+\frac{1}{10}t)^2-\frac{1}{36}s^2\Big)\nonumber\\
  &-9\ln(\frac{1}{15}-\frac{1}{15}t)-3\ln(\frac{1}{20}-\frac{1}{20}t)\nonumber\\
  &-9\ln(\frac{1}{60}+\frac{1}{90}t-\frac{1}{36}s^2)\nonumber\\
  =&-2\ln(2+3t+5s)-2\ln(2+3t-5s)\nonumber\\
  &-12\ln(1-t)-9\ln(3+2t-5s^2)+C_0, 
\end{align}
where $C_0$ is some constant.
The free energy becomes 
\begin{align}
  F(s,t)=\nu q_4(s,t)-\frac{\eta}{12}\mu_1s^2-\frac{3\eta}{20}\mu_2t^2. \label{Tengsimp}
\end{align}
It is clear that $F(s,t)=F(-s,t)$.
We investigate the phase diagram of \eqref{Tengsimp}. 

The isotropic phase corresponds to $s=t=0$. 
The case $s\ne 0$ is recognized as the tetrahedral phase; 
the case $s=0$ but $t\ne 0$ is recognized as the octahedral phase. 
By solving the local minimizers numerically, we draw the phase diagram about the two parameters $\bar{\mu}_1=\eta\mu_1/6\nu$ and $\bar{\mu}_2=3\eta\mu_2/10\nu$ (Fig. \ref{TO}).
The octahedral phase occurs when $\bar{\mu}_2$ exceeds certain value.
The appearance of tetrahedral phase depends on $\bar{\mu}_1$, for which the transition value decreases as $\bar{\mu}_2$ increases. In other words, the increase of $\bar{\mu}_2$ makes it easier for the tetrahedral phase to emerge. 
\begin{figure}
\centering
\includegraphics[width=.5\textwidth,keepaspectratio]{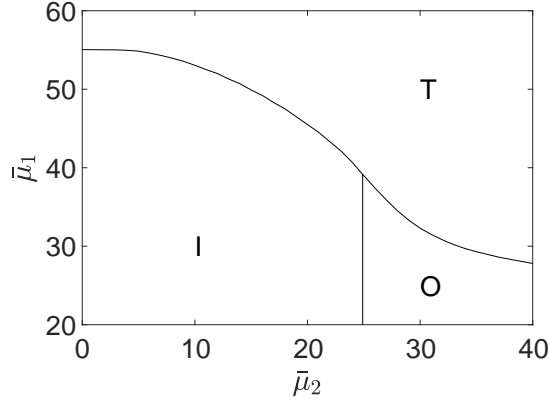}
\caption{Phase diagram of the free energy \eqref{Tengsimp} containing tetrahedral and octahedral phases. }\label{TO}
\end{figure}
To the author's knowledge, tetrahedral and octahedral phases have not been reported in one phase diagram before. 

\section{Concluding remarks\label{concl}}
We propose the quasi-entropy for averaged tensors using the log-determinant of covariance matrix. 
It is an elementary function that keeps the essential properties of the original entropy derived from constrained minimization approach: strict convexity; the ability to constrain the covariance matrix positive definite; invariant under rotations; consistency when reduced by symmetry. 
We write down the explicit expressions for several $SO(3)$ point groups: $\mathcal{D}_{\infty}$, $\mathcal{C}_{\infty}$, $\mathcal{D}_2$, $\mathcal{C}_{2}$, $\mathcal{D}_3$, $\mathcal{D}_4$, $\mathcal{T}$, and $\mathcal{O}$. 

Using the quasi-entropy, we examine some spatially homogeneous liquid crystalline systems. 
Among the results, those when adopting the quasi-entropy are consistent with previous results when adopting the original entropy. 
In particular, for rod-like molecules, we obtain the fact that at stationary points the second order tensor must have two identical eigenvalues. 
For two-fold symmetries, we prove that at stationary points the two second order tensors shall share an eigenframe under some conditions of the coefficients, and construct a class of counterexamples if these conditions do not hold. 
We also investigate the phase diagram for the bent-core molecules and molecules with tetrahedral/octahedral symmetries. 

The quasi-entropy can be derived in a unified method for all symmetries and choice of tensors, which avoids the problem of many undetermined coefficients when using polynomials as the entropy. 
Meanwhile, in contrast with the entropy derived from constrained minimization from $\int \rho\ln\rho\md\mathfrak{p}$ that involves integrals on $SO(3)$, the fact that the quasi-entropy is an elementary function is expected to greatly simplify the analysis and numerical simulations. 
It would be promising to apply the quasi-entropy to spatially inhomogeneous systems consisting of rigid molecules. 

Let us conclude with stating an open problem.
With the quasi-entropy, it is guaranteed that a certain covariance matrix $C(\avrg{\bm{w}}_{2n})$ is positive definite.
The matrix, however, only involves the components of $\langle U(\mathfrak{p})\rangle$ where $U$ is no more than $2n$-th order. 
It would be interesting to discuss whether requiring such a covariance matrix to be positive definite is equivalent to the fact that there exists a density function $\rho>0$ that gives these averaged tensors $U$.
For the case of rod-like molecules where only a $Q=\langle\bm{m}_1^2-\mathfrak{i}/3\rangle$ is involved, the statement is indeed correct.
It has been proved in \cite{ball2010nematic} that for each $Q$ with eigenvalues lying within $(-1/3,2/3)$, there exists such a density $\rho$. 
When other tensors are involved, some special cases are discussed in  \cite{BentModel}. 
To the knowledge of the author, further results of this kind are not found. 

\appendix
\section{Derivation of quasi-entropy based on $q_4^{(0)}$}
The derivation below weighs heavily on the elements of the group $\mathcal{G}$ and the expressions of invariant tensors.
We will write down the expressions we need and refer to \cite{xu_tensors} for further details. 
The Einstein convention of summation on repeated indices will be adopted.

We begin with some basic expressions. 
For each $k$-th order symmetric tensor $U$, there exists a unique $(k-2)$-th symmetric tensor $V$ such that $U-(\mathfrak{i}\otimes V)_{\mathrm{sym}}$ is a symmetric traceless tensor, which is denoted by $(U)_0$. 
Recall again that a polynomial of $\bm{m}_i$ represents a symmetric tensor, and $m_{ij}$ is the $j$-th component of $\bm{m}_i$. 
Let us review by an example, 
\begin{align*}
  (\bm{m}_1^2\bm{m}_2^2)_{ijkl}
  =&\frac{1}{6}(m_{1i}m_{1j}m_{2k}m_{2l}+m_{1i}m_{2j}m_{1k}m_{2l}+m_{1i}m_{2j}m_{2k}m_{1l}\\
  &+m_{2i}m_{1j}m_{1k}m_{2l}+m_{2i}m_{1j}m_{2k}m_{1l}+m_{2i}m_{2j}m_{1k}m_{1l}). 
\end{align*}
We write down some symmetric traceless tensors generated by monomials, 
\begin{align*}
  (\bm{m}_1^2)_0=&\bm{m}_1^2-\frac{1}{3}\mathfrak{i},\\
  (\bm{m}_1^4)_0=&\bm{m}_1^4-\frac{6}{7}\bm{m}_1^2\mathfrak{i}+\frac{3}{35}\mathfrak{i}^2, \\
  (\bm{m}_1^2\bm{m}_2^2)_0=&\bm{m}_1^2\bm{m}_2^2-\frac{1}{7}(\bm{m}_1^2+\bm{m}_2^2)\mathfrak{i}+\frac{1}{35}\mathfrak{i}^2. 
\end{align*}
Hence, we deduce that 
\begin{align*}
  \bm{m}_1^4=&(\bm{m}_1^4)_0+\frac{6}{7}\bm{m}_1^2\mathfrak{i}-\frac{3}{35}\mathfrak{i}^2\\
  =&(\bm{m}_1^4)_0+\frac{6}{7}(\bm{m}_1^2-\frac{1}{3}\mathfrak{i})\mathfrak{i}+(\frac{2}{7}-\frac{3}{35})\mathfrak{i}^2\\
  =&(\bm{m}_1^4)_0+\frac{6}{7}(\bm{m}_1^2)_0\mathfrak{i}+\frac{1}{5}\mathfrak{i}^2, \\
  \bm{m}_1^2\bm{m}_2^2=&(\bm{m}_1^2\bm{m}_2^2)_0+\frac{1}{7}(\bm{m}_1^2+\bm{m}_2^2)\mathfrak{i}-\frac{1}{35}\mathfrak{i}^2\\
  =&(\bm{m}_1^2\bm{m}_2^2)_0+\frac{1}{7}(\frac{1}{3}\mathfrak{i}-\bm{m}_3^2)\mathfrak{i}+(\frac{2}{21}-\frac{1}{35})\mathfrak{i}^2\\
  =&(\bm{m}_1^2\bm{m}_2^2)_0-\frac{1}{7}(\bm{m}_3^2)_0\mathfrak{i}+\frac{1}{15}\mathfrak{i}^2. 
\end{align*}

The Levi-Civita symbol actually defines a third order tensor, given as 
\begin{align}
  \epsilon=&\epsilon_{ijk}\bm{m}_i\otimes\bm{m}_j\otimes\bm{m}_k \nonumber\\
  =&\bm{m}_1\otimes\bm{m}_2\otimes\bm{m}_3+\bm{m}_2\otimes\bm{m}_3\otimes\bm{m}_1+\bm{m}_3\otimes\bm{m}_1\otimes\bm{m}_2\nonumber\\
  &-\bm{m}_1\otimes\bm{m}_3\otimes\bm{m}_2-\bm{m}_2\otimes\bm{m}_1\otimes\bm{m}_3-\bm{m}_3\otimes\bm{m}_2\otimes\bm{m}_1. \nonumber
\end{align}
The equality holds for any $(\bm{m}_i)$.
When there are two Levi-Civita symbols, we have 
\begin{align*}
  \epsilon_{ils}\epsilon_{jkt}=\det\left(
  \begin{array}{ccc}
    \delta_{ij} & \delta_{ik} & \delta_{it}\\
    \delta_{lj} & \delta_{lk} & \delta_{lt}\\
    \delta_{sj} & \delta_{sk} & \delta_{st}
  \end{array}
  \right). 
\end{align*}
An anti-symmetric part of a tensor can be expressed by a lower order tensor, using equalities like 
\begin{align*}
  m_{1i}m_{2j}-m_{2i}m_{1j}=\epsilon_{ijs}m_{3s}. 
\end{align*}

Now, we write down the covariance matrix of $\bm{w}_2$, for which we focus on the matrix $\langle \bm{w}_2\bm{w}_2^t\rangle$.
We divide it into four blocks, 
\begin{align*}
  \left(
  \begin{array}{cc}
    K_{11} & K_{12}\\
    K_{12}^t & K_{22}
  \end{array}
  \right), 
\end{align*}
where $K_{11}$ and $K_{22}$ are $9\times 9$ and $25\times 25$, respectively.
The block $K_{11}$ is written down in the main text. 
The blocks $K_{12}$ and $K_{22}$ are given by 
\begin{align*}
  &K_{12}=\\
  &\left(
  \begin{array}{lllll}
    \Psi_3\big(\langle \bm{m}_1 & \Psi_3\big(\frac{1}{2}\langle \bm{m}_1 & \Psi_3\big(\langle \bm{m}_1 & \Psi_3\big(\langle \bm{m}_1 & \Psi_3\big(\langle \bm{m}_1\\
    \quad \otimes(\bm{m}_1^2-\frac{1}{3}\mathfrak{i})\rangle\big) & \quad \otimes(\bm{m}_{2}^2-\bm{m}_3^2)\rangle\big) & \quad\otimes\bm{m}_1\bm{m}_2\rangle\big) & \quad\otimes\bm{m}_1\bm{m}_3\rangle\big) & \quad\otimes\bm{m}_2\bm{m}_3\rangle\big)\vspace{6pt}\\
    \Psi_3\big(\langle \bm{m}_2 & \Psi_3\big(\frac{1}{2}\langle \bm{m}_2 & \Psi_3\big(\langle \bm{m}_2 & \Psi_3\big(\langle \bm{m}_2 & \Psi_3\big(\langle \bm{m}_2\\
    \quad \otimes(\bm{m}_1^2-\frac{1}{3}\mathfrak{i})\rangle\big) & \quad \otimes(\bm{m}_{2}^2-\bm{m}_3^2)\rangle\big) & \quad\otimes\bm{m}_1\bm{m}_2\rangle\big) & \quad\otimes\bm{m}_1\bm{m}_3\rangle\big) & \quad\otimes\bm{m}_2\bm{m}_3\rangle\big)\vspace{6pt}\\
    \Psi_3\big(\langle \bm{m}_3 & \Psi_3\big(\frac{1}{2}\langle \bm{m}_3 & \Psi_3\big(\langle \bm{m}_3 & \Psi_3\big(\langle \bm{m}_3 & \Psi_3\big(\langle \bm{m}_3\\
    \quad \otimes(\bm{m}_1^2-\frac{1}{3}\mathfrak{i})\rangle\big) & \quad \otimes(\bm{m}_{2}^2-\bm{m}_3^2)\rangle\big) & \quad\otimes\bm{m}_1\bm{m}_2\rangle\big) & \quad\otimes\bm{m}_1\bm{m}_3\rangle\big) & \quad\otimes\bm{m}_2\bm{m}_3\rangle\big)
  \end{array}
  \right), 
\end{align*}
and
\begin{align*}
  &K_{22}=\\
  &\small
  \left(
  \begin{array}{lllll}
    \Psi_4\big(\langle (\bm{m}_1^2-\frac{1}{3}\mathfrak{i}) & \Psi_4\big(\frac{1}{2}\langle (\bm{m}_1^2-\frac{1}{3}\mathfrak{i}) & \Psi_4\big(\langle (\bm{m}_1^2-\frac{1}{3}\mathfrak{i}) & \Psi_4\big(\langle (\bm{m}_1^2-\frac{1}{3}\mathfrak{i}) & \Psi_4\big(\langle (\bm{m}_1^2-\frac{1}{3}\mathfrak{i})\\
    \quad \otimes(\bm{m}_1^2-\frac{1}{3}\mathfrak{i})\rangle\big) & \quad \otimes(\bm{m}_{2}^2-\bm{m}_3^2)\rangle\big) & \quad\otimes\bm{m}_1\bm{m}_2\rangle\big) & \quad\otimes\bm{m}_1\bm{m}_3\rangle\big) & \quad\otimes\bm{m}_2\bm{m}_3\rangle\big)\vspace{6pt}\\
    \Psi_4\big(\langle \frac{1}{2}(\bm{m}_2^2-\bm{m}_3^2) & \Psi_4\big(\frac{1}{4}\langle (\bm{m}_2^2-\bm{m}_3^2) & \Psi_4\big(\frac{1}{2}\langle (\bm{m}_2^2-\bm{m}_3^2) & \Psi_4\big(\frac{1}{2}\langle (\bm{m}_2^2-\bm{m}_3^2) & \Psi_4\big(\frac{1}{2}\langle (\bm{m}_2^2-\bm{m}_3^2)\\
    \quad \otimes(\bm{m}_1^2-\frac{1}{3}\mathfrak{i})\rangle\big) & \quad \otimes(\bm{m}_{2}^2-\bm{m}_3^2)\rangle\big) & \quad\otimes\bm{m}_1\bm{m}_2\rangle\big) & \quad\otimes\bm{m}_1\bm{m}_3\rangle\big) & \quad\otimes\bm{m}_2\bm{m}_3\rangle\big)\vspace{6pt}\\
    \Psi_4\big(\langle \bm{m}_1\bm{m}_2 & \Psi_4\big(\frac{1}{2}\langle \bm{m}_1\bm{m}_2 & \Psi_4\big(\langle \bm{m}_1\bm{m}_2 & \Psi_4\big(\langle \bm{m}_1\bm{m}_2 & \Psi_4\big(\langle \bm{m}_1\bm{m}_2\\
    \quad \otimes(\bm{m}_1^2-\frac{1}{3}\mathfrak{i})\rangle\big) & \quad \otimes(\bm{m}_{2}^2-\bm{m}_3^2)\rangle\big) & \quad\otimes\bm{m}_1\bm{m}_2\rangle\big) & \quad\otimes\bm{m}_1\bm{m}_3\rangle\big) & \quad\otimes\bm{m}_2\bm{m}_3\rangle\big)\vspace{6pt}\\
    \Psi_4\big(\langle \bm{m}_1\bm{m}_3 & \Psi_4\big(\frac{1}{2}\langle \bm{m}_1\bm{m}_3 & \Psi_4\big(\langle \bm{m}_1\bm{m}_3 & \Psi_4\big(\langle \bm{m}_1\bm{m}_3 & \Psi_4\big(\langle \bm{m}_1\bm{m}_3\\
    \quad \otimes(\bm{m}_1^2-\frac{1}{3}\mathfrak{i})\rangle\big) & \quad \otimes(\bm{m}_{2}^2-\bm{m}_3^2)\rangle\big) & \quad\otimes\bm{m}_1\bm{m}_2\rangle\big) & \quad\otimes\bm{m}_1\bm{m}_3\rangle\big) & \quad\otimes\bm{m}_2\bm{m}_3\rangle\big)\vspace{6pt}\\
    \Psi_4\big(\langle \bm{m}_2\bm{m}_3 & \Psi_4\big(\frac{1}{2}\langle \bm{m}_2\bm{m}_3 & \Psi_4\big(\langle \bm{m}_2\bm{m}_3 & \Psi_4\big(\langle \bm{m}_2\bm{m}_3 & \Psi_4\big(\langle \bm{m}_2\bm{m}_3\\
    \quad \otimes(\bm{m}_1^2-\frac{1}{3}\mathfrak{i})\rangle\big) & \quad \otimes(\bm{m}_{2}^2-\bm{m}_3^2)\rangle\big) & \quad\otimes\bm{m}_1\bm{m}_2\rangle\big) & \quad\otimes\bm{m}_1\bm{m}_3\rangle\big) & \quad\otimes\bm{m}_2\bm{m}_3\rangle\big)\vspace{6pt}\\
  \end{array}
  \right). 
\end{align*}
The notations $\Psi_3$ and $\Psi_4$ have been defined in the main text. 
For the four point groups, $\mathcal{O}$, $\mathcal{T}$, $\mathcal{D}_4$, $\mathcal{D}_3$, our task is to identify the vanishing tensors, and express remaining terms in the above matrices by nonvanishing symmetric traceless tensors. 


When discussing a particular group $\mathcal{G}$, we could choose to write down the expression of $q_4^{(0)}$ by symmetric traceless tensors like what is done for $q_2^{(0)}$, followed by setting the seven tensors to be zero. However, a shortcut is to calculate directly the average over a subgroup of $\mathcal{G}$ to identify some vanishing tensors, which we will adopt frequently in the following. 

Before discussing particular groups, we calculate the diagonal blocks of $K_{22}$ that are always nonzero, and three blocks in $K_{12}$ given by $\bm{m}_i\otimes\bm{m}_j\bm{m}_k$ where $i,j,k$ are mutually unequal. 
We begin with an expression to be utilized many times afterwards, 
\begin{align*}
  &2m_{2i}m_{2j}m_{3k}m_{3l}+2m_{3i}m_{3j}m_{2k}m_{2l}-(m_{2i}m_{3j}+m_{3i}m_{2j})(m_{2k}m_{3l}+m_{3k}m_{2l})\\
  =&m_{2i}m_{3l}(m_{2j}m_{3k}-m_{3j}m_{2k})+m_{2i}m_{3k}(m_{2j}m_{3l}-m_{3j}m_{2l})\\
  &+m_{3i}m_{2l}(m_{3j}m_{2k}-m_{2j}m_{2k})+m_{3i}m_{2k}(m_{3j}m_{2l}-m_{2j}m_{3l})\\
  =&(m_{2i}m_{3l}-m_{3i}m_{2l})(m_{2j}m_{3k}-m_{3j}m_{2k})+(m_{2i}m_{3k}-m_{3i}m_{2k})(m_{2j}m_{3l}-m_{3j}m_{2l})\\
  =&(\epsilon_{ils}\epsilon_{jkt}+\epsilon_{iks}\epsilon_{jlt})(\bm{m}_1^2)_{st} \\
  =&2\delta_{ij}\delta_{kl}-\delta_{ik}\delta_{jl}+\delta_{il}\delta_{jk}\\
  &-2\delta_{ij}(\bm{m}_1^2)_{kl}-2\delta_{kl}(\bm{m}_1^2)_{ij}+\delta_{ik}(\bm{m}_1^2)_{jl}+\delta_{il}(\bm{m}_1^2)_{jk}+\delta_{jk}(\bm{m}_1^2)_{il}+\delta_{jl}(\bm{m}_1^2)_{ik}\\
  =&\frac{1}{3}(2\delta_{ij}\delta_{kl}-\delta_{ik}\delta_{jl}+\delta_{il}\delta_{jk})\\
  &-2\delta_{ij}\big((\bm{m}_1^2)_0\big)_{kl}-2\delta_{kl}\big((\bm{m}_1^2)_0\big)_{ij}+\delta_{ik}\big((\bm{m}_1^2)_0\big)_{jl}+\delta_{il}\big((\bm{m}_1^2)_0\big)_{jk}+\delta_{jk}\big((\bm{m}_1^2)_0\big)_{il}+\delta_{jl}\big((\bm{m}_1^2)_0\big)_{ik}. 
\end{align*}
The eight blocks are calculated as 
\begin{align*}
  &\hspace{-36pt}(m_{1i}m_{1j}-\frac{1}{3}\delta_{ij})(m_{1k}m_{1l}-\frac{1}{3}\delta_{kl})\\
  =&\big((\bm{m}_1^4)_0+\frac{6}{7}(\bm{m}_1^2)_0\mathfrak{i}+\frac{1}{5}\mathfrak{i}^2\big)_{ijkl}
  -\frac{1}{3}\delta_{ij}\big(\bm{m}_1^2\big)_{kl}
  -\frac{1}{3}\delta_{kl}\big(\bm{m}_1^2\big)_{ij}
  +\frac{1}{9}\delta_{ij}\delta_{kl}\\
  =&\big((\bm{m}_1^4)_0+\frac{6}{7}(\bm{m}_1^2)_0\mathfrak{i}+\frac{1}{5}\mathfrak{i}^2\big)_{ijkl}
  -\frac{1}{3}\delta_{ij}\big((\bm{m}_1^2)_0\big)_{kl}
  -\frac{1}{3}\delta_{kl}\big((\bm{m}_1^2)_0\big)_{ij}
  -\frac{1}{9}\delta_{ij}\delta_{kl}\\
  =&\big((\bm{m}_1^4)_0\big)_{ijkl}
  -\frac{1}{7}\Big(\frac{4}{3}\delta_{ij}\big((\bm{m}_1^2)_0\big)_{kl}+\frac{4}{3}\delta_{kl}\big((\bm{m}_1^2)_0\big)_{ij}\\
  &-\delta_{ik}\big((\bm{m}_1^2)_0\big)_{jl}
  -\delta_{ik}\big((\bm{m}_1^2)_0\big)_{jk}
  -\delta_{jk}\big((\bm{m}_1^2)_0\big)_{il}
  -\delta_{jl}\big((\bm{m}_1^2)_0\big)_{ik}\Big)\\
  &-\frac{1}{45}(2\delta_{ij}\delta_{kl}-3\delta_{ij}\delta_{kl}-3\delta_{ij}\delta_{kl}), \\
  &\hspace{-36pt}\frac{1}{4}(m_{2i}m_{2j}-m_{3i}m_{3j})(m_{2k}m_{2l}-m_{3k}m_{3l})\\
  =&\frac{1}{4}(\bm{m}_2^4-2\bm{m}_2^2\bm{m}_3^2+\bm{m}_2^4)_{ijkl}\\
  &-\frac{1}{12}\big(2m_{2i}m_{2j}m_{3k}m_{3l}+2m_{3i}m_{3j}m_{2k}m_{2l}-(m_{2i}m_{3j}+m_{3i}m_{2j})(m_{2k}m_{3l}+m_{3k}m_{2l})\big)\\
  =&\big(\frac{1}{4}(\bm{m}_2^4+\bm{m}_3^4-2\bm{m}_2^2\bm{m}_3^2)_0+\frac{3}{14}(\bm{m}_2^2+\bm{m}_3^2)_0\mathfrak{i}+\frac{1}{14}(\bm{m}_1^2)_0\mathfrak{i}+(\frac{1}{20}+\frac{1}{20}-\frac{1}{30})\mathfrak{i}^2\big)_{ijkl}\\
  &+\frac{1}{12}\Big(2\delta_{kl}\big((\bm{m}_1^2)_0\big)_{ij}+2\delta_{ij}\big((\bm{m}_1^2)_0\big)_{kl}\\
  &\quad -\delta_{ik}\big((\bm{m}_1^2)_0\big)_{jl}-\delta_{jl}\big((\bm{m}_1^2)_0\big)_{ik}-\delta_{il}\big((\bm{m}_1^2)_0\big)_{jk}-\delta_{jk}\big((\bm{m}_1^2)_0\big)_{il}\Big)\\
  &-\frac{1}{36}(2\delta_{ij}\delta_{kl}-\delta_{ik}\delta_{jl}-\delta_{il}\delta_{jk})\\
  =&\big(\frac{1}{4}(\bm{m}_2^4+\bm{m}_3^4-2\bm{m}_2^2\bm{m}_3^2)_0\big)_{ijkl}\\
  &-\big(\frac{1}{7}(\bm{m}_1^2)_0\mathfrak{i}\big)_{ijkl}+\frac{1}{12}\Big(2\delta_{kl}\big((\bm{m}_1^2)_0\big)_{ij}+2\delta_{ij}\big((\bm{m}_1^2)_0\big)_{kl}\\
  &\quad -\delta_{ik}\big((\bm{m}_1^2)_0\big)_{jl}-\delta_{jl}\big((\bm{m}_1^2)_0\big)_{ik}-\delta_{il}\big((\bm{m}_1^2)_0\big)_{jk}-\delta_{jk}\big((\bm{m}_1^2)_0\big)_{il}\Big)\\
  &+\big(\frac{1}{15}\mathfrak{i}^2\big)_{ijkl}-\frac{1}{36}(2\delta_{ij}\delta_{kl}-\delta_{ik}\delta_{jl}-\delta_{il}\delta_{jk})\\
  =&\big(\frac{1}{4}(\bm{m}_2^4+\bm{m}_3^4-2\bm{m}_2^2\bm{m}_3^2)_0\big)_{ijkl}\\
  &+\frac{1}{7}\Big(\delta_{kl}\big((\bm{m}_1^2)_0\big)_{ij}+\delta_{ij}\big((\bm{m}_1^2)_0\big)_{kl}\\
  &\quad -\frac{3}{4}\delta_{ik}\big((\bm{m}_1^2)_0\big)_{jl}-\frac{3}{4}\delta_{jl}\big((\bm{m}_1^2)_0\big)_{ik}-\frac{3}{4}\delta_{il}\big((\bm{m}_1^2)_0\big)_{jk}-\frac{3}{4}\delta_{jk}\big((\bm{m}_1^2)_0\big)_{il}\Big)\\
  &-\frac{1}{60}(2\delta_{ij}\delta_{kl}-3\delta_{ik}\delta_{jl}-3\delta_{il}\delta_{jk}), \\
  &\hspace{-36pt}\frac{1}{4}(m_{1i}m_{2j}+m_{2i}m_{1j})(m_{1k}m_{2l}+m_{2k}m_{1l})\\
  =&(\bm{m}_1^2\bm{m}_2^2)_{ijkl}\\
  &-\frac{1}{12}\big(2m_{2i}m_{2j}m_{3k}m_{3l}+2m_{3i}m_{3j}m_{2k}m_{2l}-(m_{2i}m_{3j}+m_{3i}m_{2j})(m_{2k}m_{3l}+m_{3k}m_{2l})\big)\\
  =&\big((\bm{m}_1^2\bm{m}_2^2)_0\big)_{ijkl}
  +\frac{1}{7}\Big(\delta_{kl}\big((\bm{m}_3^2)_0\big)_{ij}+\delta_{ij}\big((\bm{m}_3^2)_0\big)_{kl}\\
  &\quad -\frac{3}{4}\delta_{ik}\big((\bm{m}_3^2)_0\big)_{jl}-\frac{3}{4}\delta_{jl}\big((\bm{m}_3^2)_0\big)_{ik}-\frac{3}{4}\delta_{il}\big((\bm{m}_3^2)_0\big)_{jk}-\frac{3}{4}\delta_{jk}\big((\bm{m}_3^2)_0\big)_{il}\Big)\\
  &-\frac{1}{60}(2\delta_{ij}\delta_{kl}-3\delta_{ik}\delta_{jl}-3\delta_{il}\delta_{jk}), \\
  &\hspace{-36pt}\frac{1}{4}(m_{1i}m_{3j}+m_{3i}m_{1j})(m_{1k}m_{3l}+m_{3k}m_{1l})\\
  =&\big((\bm{m}_1^2\bm{m}_3^2)_0\big)_{ijkl}
  +\frac{1}{7}\Big(\delta_{kl}\big((\bm{m}_2^2)_0\big)_{ij}+\delta_{ij}\big((\bm{m}_2^2)_0\big)_{kl}\\
  &\quad -\frac{3}{4}\delta_{ik}\big((\bm{m}_2^2)_0\big)_{jl}-\frac{3}{4}\delta_{jl}\big((\bm{m}_2^2)_0\big)_{ik}-\frac{3}{4}\delta_{il}\big((\bm{m}_2^2)_0\big)_{jk}-\frac{3}{4}\delta_{jk}\big((\bm{m}_2^2)_0\big)_{il}\Big)\\
  &-\frac{1}{60}(2\delta_{ij}\delta_{kl}-3\delta_{ik}\delta_{jl}-3\delta_{il}\delta_{jk}), \\
  &\hspace{-36pt}\frac{1}{4}(m_{2i}m_{3j}+m_{3i}m_{2j})(m_{2k}m_{3l}+m_{3k}m_{2l})\\
  =&\big((\bm{m}_2^2\bm{m}_3^2)_0\big)_{ijkl}
  +\frac{1}{7}\Big(\delta_{kl}\big((\bm{m}_1^2)_0\big)_{ij}+\delta_{ij}\big((\bm{m}_1^2)_0\big)_{kl}\\
  &\quad -\frac{3}{4}\delta_{ik}\big((\bm{m}_1^2)_0\big)_{jl}-\frac{3}{4}\delta_{jl}\big((\bm{m}_1^2)_0\big)_{ik}-\frac{3}{4}\delta_{il}\big((\bm{m}_1^2)_0\big)_{jk}-\frac{3}{4}\delta_{jk}\big((\bm{m}_1^2)_0\big)_{il}\Big)\\
  &-\frac{1}{60}(2\delta_{ij}\delta_{kl}-3\delta_{ik}\delta_{jl}-3\delta_{il}\delta_{jk}), \\
  &\hspace{-36pt}\frac{1}{2}m_{1i}(m_{2j}m_{3k}+m_{3j}m_{2k})\\
  =&(\bm{m}_1\bm{m}_2\bm{m}_3)_{ijk}
  +\frac{1}{6}\big((m_{1i}m_{2j}-m_{2i}m_{1j})m_{3k}+(m_{1i}m_{3j}-m_{3i}m_{1j})m_{2k}\\
  &+(m_{1i}m_{2k}-m_{2i}m_{1k})m_{3j}+(m_{1i}m_{3k}-m_{3i}m_{1k})m_{3j}\big)\\
  =&(\bm{m}_1\bm{m}_2\bm{m}_3)_{ijk}+\frac{1}{6}\Big(\epsilon_{ijs}\big((\bm{m}_3^2)_0-(\bm{m}_2^2)_0\big)_{ks}+\epsilon_{iks}\big((\bm{m}_3^2)_0-(\bm{m}_2^2)_0\big)_{js}\Big), \\
  &\hspace{-36pt}\frac{1}{2}m_{2i}(m_{1j}m_{3k}+m_{3j}m_{1k})\\
  =&(\bm{m}_1\bm{m}_2\bm{m}_3)_{ijk}+\frac{1}{6}\Big(\epsilon_{ijs}\big((\bm{m}_1^2)_0-(\bm{m}_3^2)_0\big)_{ks}+\epsilon_{iks}\big((\bm{m}_1^2)_0-(\bm{m}_3^2)_0\big)_{js}\Big), \\
  &\hspace{-36pt}\frac{1}{2}m_{3i}(m_{1j}m_{2k}+m_{2j}m_{1k})\\
  =&(\bm{m}_1\bm{m}_2\bm{m}_3)_{ijk}+\frac{1}{6}\Big(\epsilon_{ijs}\big((\bm{m}_2^2)_0-(\bm{m}_1^2)_0\big)_{ks}+\epsilon_{iks}\big((\bm{m}_2^2)_0-(\bm{m}_1^2)_0\big)_{js}\Big). 
\end{align*}
For other blocks, we will calculate when they are nonzero under certain symmetry. 




\subsection{Group $\mathcal{O}$}
The lowest order nonvanishing tensor is the fourth order tensor $O$. 
Thus, symmetric traceless tensors of third or lower order are all zero. 
However, it is very tedious if we decompose each block into symmetric traceless tensors. 
Below, we would like to make use of the symmetries directly to eliminate many blocks. 

The group allows the rotations transforming two of $\bm{m}_1$, $\bm{m}_2$, $\bm{m}_3$ into their opposites.
Therefore, only the blocks where each of $\bm{m}_1$, $\bm{m}_2$, $\bm{m}_3$ appears odd or even times can survive.
As a result, only nine blocks in $K_{12}$ and $K_{22}$ can be nonzero, of which the eight calculated above are included.
Another one is the block given by $(\bm{m}_1^2-\mathfrak{i}/3)\otimes(\bm{m}_2^2-\bm{m}_3^2)$. 

We go on examining these nine blocks.
The group allows the four-fold rotation $\mathfrak{j}_{\pi/2}$ giving $\mathfrak{pj}_{\pi/2}=(\bm{m}_1,\bm{m}_3,-\bm{m}_2)$.
Therefore, 
$$
\langle(\bm{m}_1^2-\mathfrak{i}/3)\otimes(\bm{m}_2^2-\bm{m}_3^2)\rangle
=\langle(\bm{m}_1^2-\mathfrak{i}/3)\otimes(\bm{m}_3^2-\bm{m}_2^2)\rangle
=0. 
$$
So, we only need to look into the eight blocks calculated above.

Since symmetric traceless tensors of the order no greater than three are all zero, we have 
$$
\langle(\bm{m}_1^2)_0\rangle=\langle(\bm{m}_2^2)_0\rangle=\langle(\bm{m}_3^2)_0\rangle=\langle\bm{m}_1\bm{m}_2\bm{m}_3\rangle=0. 
$$
Furthermore, the group allows the rotations
\begin{align}
(\bm{m}_1, \bm{m}_2, \bm{m}_3)\to (\bm{m}_2, \bm{m}_3, \bm{m}_1), (\bm{m}_3, \bm{m}_1, \bm{m}_2). \label{3fold}
\end{align}
Together with the definition of the tensor $O$, 
\begin{align*}
  O=&\langle\bm{m}_1^2\bm{m}_2^2+\bm{m}_2^2\bm{m}_3^2+\bm{m}_1^2\bm{m}_3^2-\frac{1}{5}\mathfrak{i}^2\rangle\\
  =&\langle(\bm{m}_1^2\bm{m}_2^2+\bm{m}_2^2\bm{m}_3^2+\bm{m}_1^2\bm{m}_3^2)_0\rangle\\
  =&-\frac{1}{2}\langle(\bm{m}_1^4+\bm{m}_2^4+\bm{m}_3^4)_0\rangle, 
\end{align*}
we deduce that 
\begin{align*}
  &\langle(\bm{m}_1^2\bm{m}_2^2)_0\rangle=\langle(\bm{m}_2^2\bm{m}_3^2)_0\rangle=\langle(\bm{m}_1^2\bm{m}_3^2)_0\rangle\\
  &\qquad =-\frac{1}{2}\langle(\bm{m}_1^4)_0\rangle=-\frac{1}{2}\langle(\bm{m}_2^4)_0\rangle=-\frac{1}{2}\langle(\bm{m}_3^4)_0\rangle=\frac{1}{3}O. 
\end{align*}
Taking them into the expressions of the eight blocks, we obtain the quasi-entropy. 
Here, we note that 
\begin{align*}
  \frac{1}{60}\Psi_4(3\delta_{ik}\delta_{jl}+3\delta_{il}\delta_{jk}-2\delta_{ij}\delta_{kl})=\mathrm{diag}(\frac{1}{15},\frac{1}{20},\frac{1}{20},\frac{1}{20},\frac{1}{20})=D 
\end{align*}
is the diagonal matrix defined in \eqref{Dmat}. 

\subsection{Group $\mathcal{T}$}
The lowest order nonvanishing tensor is $T=\langle\bm{m}_1\bm{m}_1\bm{m}_3\rangle$.
Thus, the symmetric traceless tensors of first or second order are all zero. 
Most derivations for the group $\mathcal{O}$ are still valid. 
The only difference is that the block given by $(\bm{m}_1^2-\frac{1}{3}\mathfrak{i})\otimes(\bm{m}_2^2-\bm{m}_3^2)$ is nonzero. 
Since the second order symmetric traceless tensor $\langle\bm{m}_2^2-\bm{m}_3^2\rangle=0$, we calculate the average as 
\begin{align*}
  \frac{1}{2}\langle(\bm{m}_1^2-\frac{1}{3}\mathfrak{i})_{ij}(\bm{m}_2^2-\bm{m}_3^2)_{kl}\rangle
  =\frac{1}{2}\langle(\bm{m}_1^2)_{ij}(\bm{m}_2^2-\bm{m}_3^2)_{kl}\rangle. 
\end{align*}
By the three-fold rotations \eqref{3fold}, we have 
\begin{align*}
  \langle(\bm{m}_1^2)_{ij}(\bm{m}_2^2)_{kl}\rangle=\langle(\bm{m}_2^2)_{ij}(\bm{m}_3^2)_{kl}\rangle=\langle(\bm{m}_3^2)_{ij}(\bm{m}_1^2)_{kl}\rangle.
\end{align*}
Therefore, 
\begin{align*}
  &\frac{1}{2}\langle(\bm{m}_1^2)_{ij}(\bm{m}_2^2-\bm{m}_3^2)_{kl}\rangle\\
  =&\frac{1}{6}\langle(\bm{m}_1^2)_{ij}(\bm{m}_2^2)_{kl}+(\bm{m}_2^2)_{ij}(\bm{m}_3^2)_{kl}+(\bm{m}_3^2)_{ij}(\bm{m}_1^2)_{kl}\\
  &-(\bm{m}_2^2)_{ij}(\bm{m}_1^2)_{kl}-(\bm{m}_3^2)_{ij}(\bm{m}_2^2)_{kl}-(\bm{m}_1^2)_{ij}(\bm{m}_3^2)_{kl}\rangle. 
\end{align*}
Using
\begin{align*}
  &(\bm{m}_1^2)_{ij}(\bm{m}_2^2)_{kl}-(\bm{m}_2^2)_{ij}(\bm{m}_1^2)_{kl}\\
  =&m_{1i}m_{1j}m_{2k}m_{2l}-m_{1k}m_{1j}m_{2i}m_{2l}+m_{1k}m_{1j}m_{2i}m_{2l}-m_{1k}m_{1l}m_{2i}m_{2j}\\
  =&\epsilon_{iks}m_{1j}m_{2l}m_{3s}+\epsilon_{jls}m_{1k}m_{2i}m_{3s}, 
\end{align*}
we arrive at 
\begin{align*}
  \frac{1}{6}\langle&(\bm{m}_1^2)_{ij}(\bm{m}_2^2)_{kl}+(\bm{m}_2^2)_{ij}(\bm{m}_3^2)_{kl}+(\bm{m}_3^2)_{ij}(\bm{m}_1^2)_{kl}\\
  &-(\bm{m}_2^2)_{ij}(\bm{m}_1^2)_{kl}-(\bm{m}_3^2)_{ij}(\bm{m}_2^2)_{kl}-(\bm{m}_1^2)_{ij}(\bm{m}_3^2)_{kl}\rangle\\
  =&\frac{1}{6}\epsilon_{iks}\langle m_{1j}m_{2l}m_{3s}+m_{2j}m_{3l}m_{1s}+m_{3j}m_{1l}m_{2s}\rangle\\
  &+\frac{1}{6}\epsilon_{jls}\langle m_{1k}m_{2i}m_{3s}+m_{2k}m_{3i}m_{1s}+m_{3k}m_{1i}m_{2s}\rangle\\
  =&\frac{1}{2}\epsilon_{iks}\langle \bm{m}_1\bm{m}_2\bm{m}_3\rangle_{jls}+\frac{1}{2}\epsilon_{jls}\langle\bm{m}_1\bm{m}_2\bm{m}_3\rangle_{kis}\\
  &+\frac{1}{12}\epsilon_{iks}\langle m_{1j}m_{2l}m_{3s}+m_{2j}m_{3l}m_{1s}+m_{3j}m_{1l}m_{2s}-m_{1j}m_{3l}m_{2s}-m_{2j}m_{1l}m_{3s}-m_{3j}m_{2l}m_{1s}\rangle\\
  &+\frac{1}{12}\epsilon_{jls}\langle m_{1k}m_{2i}m_{3s}+m_{2k}m_{3i}m_{1s}+m_{3k}m_{1i}m_{2s}-m_{1k}m_{3i}m_{2s}-m_{2k}m_{1i}m_{3s}-m_{3k}m_{2i}m_{1s}\rangle\\
  =&\frac{1}{2}\epsilon_{iks}T_{jls}+\frac{1}{2}\epsilon_{jls}T_{kis}
  +\frac{1}{12}(\epsilon_{iks}\epsilon_{jls}+\epsilon_{jls}\epsilon_{kis})\\
  =&\frac{1}{2}\epsilon_{iks}T_{jls}+\frac{1}{2}\epsilon_{jls}T_{kis}. 
\end{align*}
Here, we notice that for any third order tensor $V$, the fourth order tensor $\epsilon_{iks}V_{jls}+\epsilon_{jls}V_{kis}$ is symmetric about the index pairs $(i,j)$ and $(k,l)$, as long as $V_{iki}=V_{kjj}=0$. Actually, 
\begin{align*}
  \epsilon_{ijt}(\epsilon_{iks}V_{jls}+\epsilon_{jls}V_{kis})
  =&V_{ili}\delta_{kt}-V_{kjj}\delta_{lt}. 
\end{align*}
As a special case, when $V$ is symmetric traceless, we have 
\begin{align*}
  \epsilon_{iks}V_{jls}+\epsilon_{jls}V_{kis}=\epsilon_{jks}V_{ils}+\epsilon_{ils}V_{kjs}. 
\end{align*}

\subsection{Group $\mathcal{D}_4$}
Following the same derivation for the group $\mathcal{O}$, we deduce that only the eight blocks calculated at the beginning of the section are nonzero. 
Rewrite the nonvanishing tensors as
\begin{align*}
  &Q^2=\langle(\bm{m}_1^2)_0\rangle,\qquad
  Q^4=\langle(\bm{m}_1^4)_0\rangle,\\
  &M^4_1=\langle(8\bm{m}_2^4-8(I-\bm{m}_1^2)\bm{m}_2^2+(I-\bm{m}_1^2)^2\rangle=\langle(8\bm{m}_2^4+8\bm{m}_1^2\bm{m}_2^2+\bm{m}_1^4)_0\rangle. 
\end{align*}
We then calculate using $\bm{m}_1^2+\bm{m}_2^2+\bm{m}_3^2=\mathfrak{i}$ that 
\begin{align*}
  &\langle(\bm{m}_2^2)_0\rangle=\langle(\bm{m}_3^2)_0\rangle=-\frac{1}{2}Q^2, \\
  &\langle(\bm{m}_1^2\bm{m}_2^2)_0\rangle=\langle(\bm{m}_1^2\bm{m}_3^2)_0\rangle=-\frac{1}{2}Q^4.
\end{align*}
Then, using the four-fold rotation $\mathfrak{j}_{\pi/2}$, we deduce that 
\begin{align*}
  &\langle(\bm{m}_2^4)_0\rangle=\langle(\bm{m}_3^4)_0\rangle=\frac{1}{8}M^4_1-\langle \bm{m}_1^2\bm{m}_2^2\rangle-\frac{1}{8}\langle\bm{m}_1^4\rangle=\frac{1}{8}M^4_1+\frac{3}{8}Q^4, \\
  &\langle(\bm{m}_2^2\bm{m}_3^2)_0\rangle=-\langle(\bm{m}_1^2\bm{m}_2^2)_0\rangle-\langle(\bm{m}_2^4)_0\rangle=\frac{1}{8}(Q^4-M^4_1). 
\end{align*}

\subsection{Group $\mathcal{D}_3$}
For the eight blocks calculated at the beginning, they are identical to the group $\mathcal{D}_4$ where we need to set $M^4_1=0$. 

There are four extra nonzero blocks, given by $\bm{m}_3\otimes\bm{m}_2\bm{m}_3$, $\frac{1}{2}\bm{m}_2\otimes(\bm{m}_2^2-\bm{m}_3^2)$, $\bm{m}_1\bm{m}_2\otimes\bm{m}_2\bm{m}_3$, and $\frac{1}{2}(\bm{m}_2^2-\bm{m}_3^2)\otimes\bm{m}_1\bm{m}_3$. 
For these blocks, we need to examine the average on the three-fold rotations, $\mathfrak{p}=(\bm{m}_1,\bm{m}_2,\bm{m}_3)$ and
\begin{align*}
  &\mathfrak{pj}_{2\pi/3}=(\bm{m}_1, -\frac{1}{2}\bm{m}_2+\frac{\sqrt{3}}{2}\bm{m}_3, -\frac{\sqrt{3}}{2}\bm{m}_2-\frac{1}{2}\bm{m}_3),\\
  &\mathfrak{pj}_{4\pi/3}=(\bm{m}_1,-\frac{1}{2}\bm{m}_2-\frac{\sqrt{3}}{2}\bm{m}_3, \frac{\sqrt{3}}{2}\bm{m}_2-\frac{1}{2}\bm{m}_3). 
\end{align*}

We start from the block $\bm{m}_3\otimes\bm{m}_2\bm{m}_3$, for which we have 
\begin{align*}
  &\hspace{-36pt}\frac{1}{2}m_{3i}(m_{2j}m_{3k}+m_{3j}m_{2k})\\
  =&(\bm{m}_2\bm{m}_3^2)_{ijk}+\frac{1}{6}(m_{2i}m_{3j}-m_{3i}m_{2j})m_{3k}+\frac{1}{6}(m_{2i}m_{3k}-m_{3i}m_{2k})m_{3j}\\
  =&(\bm{m}_2\bm{m}_3^2)_{ijk}+\frac{1}{6}\epsilon_{ijs}m_{1s}m_{3k}+\frac{1}{6}\epsilon_{iks}m_{1s}m_{3j}. 
\end{align*}
Using the three-fold rotations, we have
\begin{align*}
  \langle m_{1s}m_{3k}\rangle=\frac{1}{3}\big(\langle m_{1s}m_{3k}\rangle+\langle m_{1s}(-\frac{\sqrt{3}}{2}m_{2k}-\frac{1}{2}m_{3k})\rangle+\langle m_{1s}(\frac{\sqrt{3}}{2}m_{2k}-\frac{1}{2}m_{3k})\rangle\big)=0. 
\end{align*}
The first term yields 
\begin{align*}
  \bm{m}_2\bm{m}_3^2=&-\bm{m}_2^3-\bm{m}_1^2\bm{m}_2+\mathfrak{i}\bm{m}_2\\
  =&-\big(\bm{m}_2^3-\frac{3}{4}(\mathfrak{i}-\bm{m}_1^2)\bm{m}_2\big)
  +\frac{1}{4}(\mathfrak{i}-\bm{m}_1^2)\bm{m}_2. 
\end{align*}
When averaged, the second term is zero, and the first term gives $-\frac{1}{4}M_1^3$. Therefore, we arrive at 
\begin{align*}
  \langle\frac{1}{2}m_{3i}(m_{2j}m_{3k}+m_{3j}m_{2k})\rangle=-\frac{1}{4}(M_1^3)_{ijk}. 
\end{align*}
Similarly, we deduce that 
\begin{align*}
  \frac{1}{2}\langle m_{2i}(m_{2j}m_{2k}-m_{3j}m_{3k})\rangle
  =&\frac{1}{2}\langle m_{2i}(2m_{2j}m_{2k}+m_{1j}m_{1k}-\delta_{jk})\rangle
  =\langle \bm{m}_2^3\rangle_{ijk}
  =\frac{1}{4}(M_1^3)_{ijk}. 
\end{align*}

For the block $\bm{m}_1\bm{m}_2\otimes\bm{m}_2\bm{m}_3$, we have 
\begin{align*}
  m_{1i}m_{2j}m_{2k}m_{3l}
  =&\delta_{jk}m_{1i}m_{3l}-m_{1i}m_{1j}m_{1k}m_{3l}-m_{1i}m_{3j}m_{3k}m_{3l}\\
\end{align*}
The averages of the first two terms are zero because of three-fold rotations.
So, we obtain 
\begin{align*}
  &\frac{1}{4}\langle(m_{1i}m_{2j}+m_{2i}m_{1j})(m_{2k}m_{3l}+m_{3k}m_{2l})\rangle\\
  =&\frac{1}{2}\langle -m_{1i}m_{3j}m_{3k}m_{3l}-m_{3i}m_{1j}m_{3k}m_{3l}\rangle\\
  =&-\langle\bm{m}_1\bm{m}_3^3\rangle_{ijkl}+\frac{1}{4}\langle m_{3i}m_{3j}m_{1k}m_{3l}+m_{3i}m_{3j}m_{3k}m_{1l}-m_{1i}m_{3j}m_{3k}m_{3l}-m_{3i}m_{1j}m_{3k}m_{3l}\rangle\\
  =&-\langle\bm{m}_1\bm{m}_3^3\rangle_{ijkl}+\frac{1}{4}(\epsilon_{iks}\langle m_{2s}m_{3j}m_{3l}\rangle+\epsilon_{jls}\langle m_{2s}m_{3i}m_{3k}\rangle). 
\end{align*}
The tensor $\bm{m}_1\bm{m}_3^3$ can be expressed as 
\begin{align*}
  -\bm{m}_1\bm{m}_3^3=&-\mathfrak{i}\bm{m}_1\bm{m}_3+\bm{m}_1^3\bm{m}_3+\bm{m}_1\bm{m}_2^2\bm{m}_3\\
  =&\bm{m}_1(\bm{m}_2^2-\frac{1}{4}(\mathfrak{i}-\bm{m}_1^2))\bm{m}_3
  -\frac{3}{4}\mathfrak{i}\bm{m}_1\bm{m}_3+\frac{3}{4}\bm{m}_1^3\bm{m}_3. 
\end{align*}
When averaged, the first term gives $\frac{1}{4}N^4$, while the others are zero because of three-fold rotations. 
Therefore,
\begin{align*}
  \frac{1}{4}\langle(m_{1i}m_{2j}+m_{2i}m_{1j})(m_{2k}m_{3l}+m_{3k}m_{2l})\rangle
  =\frac{1}{4}(N^4)_{ijkl}-\frac{1}{4}\big(\epsilon_{iks}(M_1^3)_{jls}+\epsilon_{jls}(M_1^3)_{iks}\big). 
\end{align*}

For the block $\frac{1}{2}(\bm{m}_2^2-\bm{m}_3^2)\otimes \bm{m}_1\bm{m}_3$, we can similarly calculate 
\begin{align*}
  &\frac{1}{4}\langle(m_{2i}m_{2j}-m_{3i}m_{3j})(m_{1k}m_{3l}+m_{3k}m_{1l})\rangle\\
  =&\frac{1}{4}\langle(\delta_{ij}-m_{1i}m_{1j}-2m_{3i}m_{3j})(m_{1k}m_{3l}+m_{3k}m_{1l})\rangle\\
  =&-\frac{1}{2}\langle m_{3i}m_{3j}(m_{1k}m_{3l}+m_{3k}m_{1l})\rangle\\
  =&-\langle\bm{m}_1\bm{m}_3^3\rangle_{ijkl}-\frac{1}{4}\langle m_{3i}m_{3j}m_{1k}m_{3l}+m_{3i}m_{3j}m_{3k}m_{1l}-m_{1i}m_{3j}m_{3k}m_{3l}-m_{3i}m_{1j}m_{3k}m_{3l}\rangle\\
  =&-\langle\bm{m}_1\bm{m}_3^3\rangle_{ijkl}-\frac{1}{4}(\epsilon_{iks}\langle m_{2s}m_{3j}m_{3l}\rangle+\epsilon_{jls}\langle m_{2s}m_{3i}m_{3k}\rangle)\\
  =&\frac{1}{4}(N^4)_{ijkl}+\frac{1}{4}\big(\epsilon_{iks}(M_1^3)_{jls}+\epsilon_{jls}(M_1^3)_{iks}\big). 
\end{align*}

\bibliographystyle{plain}
\bibliography{bib_sym}

\end{document}